\documentclass[a4paper,onecolumn,11pt,accepted=2026-01-06]{quantumarticle}
\pdfoutput=1
\usepackage[utf8]{inputenc}
\usepackage[english]{babel}
\usepackage[T1]{fontenc}
\usepackage[numbers]{natbib}
\usepackage{amsthm, amsfonts, amsmath, amssymb}
\usepackage{mathtools}
\usepackage{thmtools}
\usepackage{mathrsfs}
\usepackage{braket}
\usepackage{enumerate}
\usepackage{enumitem}
\usepackage[ruled]{algorithm2e}
\usepackage{hyperref}
\usepackage{csquotes}

\usepackage[capitalize]{cleveref}
\crefname{thm}{Theorem}{Theorems}
\crefname{lem}{Lemma}{Lemmas}

\theoremstyle{plain}
\newtheorem{thm}{Theorem}[section]
\newtheorem*{thm*}{Theorem}
\newtheorem{prop}[thm]{Proposition}
\newtheorem{cor}[thm]{Corollary}
\newtheorem{lem}[thm]{Lemma}
\newtheorem*{lem*}{Lemma}
\newtheorem{prob}[thm]{Problem}

\theoremstyle{remark}
\newtheorem{note}{Remark}[thm]

\theoremstyle{definition}

\newtheorem{definition}[thm]{Definition}

\newtheorem*{note*}{Remark}

\DeclareMathOperator{\rel}{rel}
\DeclareMathOperator{\tr}{Tr}
\DeclareMathOperator{\soe}{soe}
\DeclareMathOperator{\iso}{ISO}
\DeclareMathOperator{\spn}{span}

\newcommand{\card}[1]{\left\lvert #1 \right\rvert}
\newcommand{\complex}{\mathbb{C}}
\def\calA{{\mathcal A}} \def\calB{{\mathcal B}} 
  \def\calF{{\mathcal F}}
 \def\calH{{\mathcal H}} 
  \def\calL{{\mathcal L}}
  
\def\calP{{\mathcal P}} \def\calQ{{\mathcal Q}} \def\calR{{\mathcal R}}
\def\calS{{\mathcal S}} \def\calT{{\mathcal T}}

\newcommand{\cyclicpermutations}{\mathscr{C}(1,\dots, k, 2k, \dots, k+1)}

\renewcommand{\epsilon}{\varepsilon}

\setlist[enumerate, 1]{font=\upshape, noitemsep, nolistsep}
\setlist[enumerate, 2]{font=\upshape, noitemsep, nolistsep}
\setlist[itemize, 1]{noitemsep, nolistsep,font=\upshape}
\setlist[itemize, 2]{noitemsep, nolistsep,font=\upshape}

\begin{document}

\title{NPA Hierarchy for Quantum Isomorphism and\texorpdfstring{\newline}{ } Homomorphism Indistinguishability$^{*}$}
\author{Prem Nigam Kar}
\affiliation{Technical University of Denmark}
\affiliation{Instytut Matematyczny Polskiej Akademii Nauk}
\homepage{https://sites.google.com/view/prem-nigam-kar/}
\thanks{Carlsberg Semper Ardens Accelerate CF21-0682 Quantum Graph Theory.}
\orcid{0009-0000-8235-9233}
\author{David E. Roberson}
\affiliation{Technical University of Denmark}
\affiliation{QMATH, University of Copenhagen}
\homepage{https://sites.google.com/site/davideroberson/}
\orcid{0000-0002-4463-8095}
\author{Tim Seppelt}
\affiliation{IT-Universitetet i København}
\homepage{https://tseppelt.github.io}
\thanks{German Research Foundation (DFG) within Research Training Group 2236/2 (UnRAVeL) and European Union (ERC, SymSim, 101054974, CountHom, 101077083)}
\orcid{0000-0002-6447-0568}
\author{Peter Zeman}
\affiliation{Department of Algebra, Faculty of Mathematics and Physics, Charles University}
\homepage{https://zemanpeter.matfyz.cz}
\thanks{Funded by the European Union (ERC, POCOCOP, 101071674).}
\orcid{0000-0003-0071-9149}
\thanks{\texorpdfstring{\newline}{ }\texorpdfstring{\newline}{ }* Extended abstract was accepted to the Proceedings of 52nd EATCS International Colloquium on Automata, Languages, and Programming (ICALP) 2025. Views and opinions expressed are however those of the author(s) only and do not necessarily reflect those of the European Union or the European Research Council Executive Agency. Neither the European Union nor the granting authority can be held responsible for them.}

\maketitle

\begin{abstract}

    Mančinska and Roberson~[FOCS'20] showed that two graphs are quantum isomorphic if and only if they admit the same number of homomorphisms from any planar graph. Atserias et al.~[JCTB'19] proved that quantum isomorphism is undecidable in general, which motivates the study of its relaxations. In the classical setting, Roberson and Seppelt~[ICALP'23] characterized the feasibility of each level of the Lasserre hierarchy of semidefinite programming relaxations of graph isomorphism in terms of equality of homomorphism counts from an appropriate graph class. The NPA hierarchy, a noncommutative generalization of the Lasserre hierarchy, provides a sequence of semidefinite programming relaxations for quantum isomorphism. In the quantum setting, we show that the feasibility of each level of the NPA hierarchy for quantum isomorphism is equivalent to equality of homomorphism counts from an appropriate class of planar graphs. Combining this characterization with the convergence of the NPA hierarchy, and noting that the union of these classes is the set of all planar graphs, we obtain a new proof of the result of Mančinska and Roberson~[FOCS'20] that avoids the use of quantum groups. Moreover, this homomorphism indistinguishability characterization also yields a randomized polynomial-time algorithm deciding exact feasibility of each fixed level of the NPA hierarchy of SDP relaxations for quantum isomorphism.
    
\end{abstract}

\newpage
\setcounter{tocdepth}{2}
\tableofcontents

\section{Introduction}

Two graphs $G$ and $H$ are said to be \emph{homomorphism indistinguishable} over a class of graphs $\calF$, written $G \cong_{\calF} H$, if for every graph $F \in \calF$, the number of homomorphisms from $F$ to $G$ is the same as the number of homomorphisms from $F$ to $H$. A classic result from \cite{lovasz_operations_1967} states that two graphs $G$ and $H$ are isomorphic if and only if they are homomorphism indistinguishable over all graphs. 
Since then, several relaxations of graph isomorphism from fields as diverse as finite model theory \cite{sherali-adams-3,grohe_counting_2020,fluck_going_2024}, algebraic graph theory \cite{dell_lovasz_2018}, optimisation \cite{grohe_homomorphism_2022,roberson-seppelt-arxiv}, machine learning \cite{xu_how_2019,morris_weisfeiler_2019,zhang_beyond_2024}, and category theory \cite{dawar_lovasz-type_2021,abramsky_discrete_2022,montacute_pebble-relation_2022} were found to be homomorphism indistinguishability relations over suitable graph classes. 
Recently, a coherent theory which addresses the descriptive and computational complexity of homomorphism indistinguishability relations has begun to emerge \cite{roberson_oddomorphisms_2022,seppelt_logical_2023,neuen_homomorphism-distinguishing_2023,seppelt_algorithmic_2024}, cf.\ the monograph \cite{seppelt_homomorphism_2024}.

A ground-breaking connection between quantum information and homomorphism indistinguishability was found by Man\v{c}inska and Roberson \cite{david-laura}: 
They showed that two graphs are quantum isomorphic if and only if they are homomorphism indistinguishable over all planar graphs.
Quantum isomorphism, as introduced by \cite{ATSERIAS2019289}, is a natural relaxation of graph isomorphism in terms of the graph isomorphism game, where two cooperating players called Alice and Bob try to convince a referee that two graphs are isomorphic. A perfect deterministic strategy exists for the $(G, H)$-isomorphism game if and only if the graphs $G$ and $H$ are isomorphic. 
Two graphs $G$ and $H$ are said to be \emph{quantum isomorphic}, written $G \cong_q H$, if there is a perfect quantum strategy $(G, H)$-isomorphism game, i.e., a perfect strategy making use of local quantum measurements on a shared entangled state. 

The proof of Man\v{c}inska's and Roberson's result \cite{david-laura} heavily relies on certain esoteric mathematical objects called quantum groups. 
In more detail, the main idea of the proof is to interpret homomorphism tensors of bilabelled graphs as intertwiners of the quantum automorphism groups of the respective graphs. 

Another recent result from \cite{roberson-seppelt-arxiv} also obtained a homomorphism indistinguishability characterisation for each level of the Lasserre hierarchy of semidefinite programming (SDP) relaxations for the integer program for isomorphism between two graphs. 
More precisely, for each $k \in \mathbb{N}$, the authors of \cite{roberson-seppelt-arxiv} constructed a class of graphs $\calL_k$ such that the $k^{\text{th}}$-level of the Lasserre hierarchy of SDP relaxations of the integer program for deciding whether $G$ and $H$ are isomorphic is feasible if and only if $G$ and $H$ are homomorphism indistinguishable over the graph class $\calL_k$. 

It was also shown in \cite{ATSERIAS2019289} that deciding quantum isomorphism of graphs is undecidable---contrary to deciding isomorphism of graphs, which is clearly decidable and currently known to be solvable in quasipolynomial time \cite{babai_graph_2016}. 
This motivates the study of relaxations of quantum isomorphism. 
The NPA hierarchy \cite{Navascues_2008}, which can be thought of as a noncommutative analogue of the Lasserre hierarchy, is a sequence of SDP relaxations of the problem of determining if a given joint conditional probability distribution arises from quantum mechanics. In particular, the NPA hierarchy can be used to formulate a hierarchy of SDP relaxations for the problem of deciding whether two graphs are quantum isomorphic. 

In light of results from \cite{david-laura,roberson-seppelt-arxiv}, 
it is natural to ask if the feasibility of each level of these SDP relaxations of quantum isomorphism can be characterised as a homomorphism indistinguishability relation over some family of graphs. 
Such a characterisation would make the NPA hierarchy subject to the emerging theory of homomorphism indistinguishability.
For example, a recent result \cite{seppelt_algorithmic_2024} asserts that, over every minor-closed graph class of bounded treewidth, homomorphism indistinguishability can be decided in randomized polynomial time.
The NPA relaxation, being a semidefinite program, can be solved using standard techniques such as the ellipsoid method. However, such techniques can, in polynomial time, only decide the approximate feasibility of a system.
A homomorphism indistinguishability characterisation of the NPA hierarchy would imply, for each of its levels, the existence of a randomized polynomial-time algorithm for deciding exact feasibility \cite{neuen_homomorphism-distinguishing_2023,seppelt_logical_2023,seppelt_algorithmic_2024}.

\subsection{Main Results}

Our main contribution is a homomorphism indistinguishability characterization for each level of the NPA hierarchy, as formalized by the following theorem.
\begin{thm}[\textbf{Main Theorem}]\label{thm:main-theorem}
    For graphs $G$ and $H$ and $k \in \mathbb{N}$, the following are equivalent:
    \begin{enumerate}[label = (\roman*)]
        \item there is a solution for the $k^{\text{th}}$-level of the NPA hierarchy for the $(G,H)$-isomorphism game;
        \item there is a level-$k$ quantum isomorphism map from $G$ to $H$;\label{ref1}
        \item $G$ and $H$ are algebraically $k$-equivalent;\label{ref2}
        \item $G$ and $H$ are homomorphism indistinguishable over the family of graphs $\calP_k$. 
    \end{enumerate}
\end{thm}
In particular, the $k^{\text{th}}$-level of the NPA hierarchy is feasible for the $(G,H)$-isomorphism game if and only $G$ and $H$ are homomorphism indistinguishable over the graph class $\calP_k$. 
Here, $\mathcal{P}_k$ is a bounded-treewidth minor-closed class of planar graphs, which we construct by interpreting the NPA systems of equations in light of a correspondence between combinatorics (bilabelled graphs) and algebra (homomorphism tensors) which underpins many recent results regarding homomorphism indistinguishability \cite{david-laura,grohe_homomorphism_2022,rattan_weisfeiler_2023,roberson-seppelt-arxiv}.

As a corollary of \cref{thm:main-theorem}, we devise
a randomized polynomial-time algorithm for deciding the exact feasibility of each level of the NPA hierarchy.
To that end, we first show that the graph classes $\calP_k$ are minor-closed and of bounded treewidth, which is a graph parameter roughly measuring how far is a graph from a tree.
Hence, a recent result from \cite{seppelt_algorithmic_2024} implies that, for each $k \in \mathbb{N}$, there exists a randomized polynomial-time algorithm for deciding homomorphism indistinguishability over~$\mathcal{P}_k$.
We strengthen this result by making the dependence on the parameter~$k$ effective.
%

\begin{thm}[restate=thmAlgorithm, label=thm:algorithmic-aspects]
    There exists a randomized algorithm which decides, given graphs $G$ and $H$ and an integer $k\geq 1$,
    whether the $k^{\text{th}}$-level of the NPA hierarchy for the $(G,H)$-isomorphism game is feasible.
    The algorithm always runs in time $n^{O(k)} k^{O(1)}$ for $n \coloneqq \max \{\lvert V(G) \rvert , \lvert V(H) \rvert  \}$, accepts all YES-instances, and accepts NO-instances with probability less than one half.
\end{thm}

\subsection{Proof Techniques}

The main algebraic-combinatorial tools we use are bilabelled graphs and their homomorphism tensors. \emph{Bilabelled graphs} are graphs with distinguished vertices which are said to carry labels. To a bilabelled graph $\boldsymbol{F} = (F, u, v)$ and an unlabelled graph $G$, one can associated the \emph{homomorphism tensor} $\boldsymbol{F}_G \in \mathbb{N}^{V(G) \times V(G)}$ such that $\boldsymbol{F}_G(x, y)$ for $x,y \in V(G)$ is the number of homomorphisms $h \colon F \to G$ such that $h(u) = x$ and $h(v) = y$.
For example, the bilabelled graph $\boldsymbol{A} = (A, u, v)$ with $V(A) = \{u,v\}$ and $E(A) = \{uv\}$ denotes the complete $2$-vertex graph each of whose vertices $u$ and $v$ carry one label.
In the case of $\boldsymbol{A}$, the matrix $\boldsymbol{A}_G$ is just the adjacency matrix of $G$.
The fruitfulness of this construction stems from a correspondence between combinatorial operations on bilabelled graphs and algebraic operations on homomorphism tensors.
For example, the \emph{matrix product} $(\boldsymbol{F}_1)_G \cdot (\boldsymbol{F}_2)_G$ yields the homomorphism tensor of the bilabelled graph obtained by taking the \emph{series composition} of $\boldsymbol{F}_1$ and $\boldsymbol{F}_2$.

We prove \cref{thm:main-theorem} by interpreting the NPA relaxation as a system of equations whose constraints involve homomorphism tensors and algebraic operations.
Applying the aforementioned algebro-combinatorial correspondence,
the graph class $\mathcal{P}_k$ is then obtained by reading the constraints as a description of a graph class via bilabelled graphs and combinatorial operations.
To that end, we first give various reformulations of the NPA systems of equations as listed in \cref{ref1,ref2} of \cref{thm:main-theorem}.
The proofs follow the outline below:

\begin{itemize}
    \item In \cref{thm:main-theorem-1}, we first show that a principal submatrix of a certificate for the $k^{\text{th}}$-level of the NPA hierarchy for the $(G,H)$-isomorphism game can be interpreted as a the Choi matrix of a completely positive map from $M_{V(G)^k}(\complex)$ to $M_{V(H)^k}(\complex)$ with certain properties. Such a completely positive map is known as a level-$k$ quantum isomorphism map. We also show that the Choi matrix of such a level-$k$ quantum isomorphism map (uniquely) extends to a certificate for the $k^{\text{th}}$-level of the NPA hierarchy for quantum isomorphism, thus showing that the existence of such a map is equivalent to the feasibility of the $k^{\text{th}}$-lvel of the NPA hierarchy.  

    \item In \cref{thm:main-theorem-2}, restrictions of the aforementioned completely positive maps are then shown to be algebra homomorphisms mapping homomorphism tensors of $\calQ_k$ for $G$ to homomorphism tensors of $\calQ_k$ for $H$, where $\calQ_k$ is the set of \emph{atomic graphs} which form the building blocks for the graph class $\calP_k$. Such an algebra homomorphism is called an \emph{algebraic $k$-equivalence.}

    \item Lastly, in \cref{thm:main-theorem-3}, we use the correspondence between combinatorial operations on graphs and algebraic operations on their homomorphism tensors to show that the existence of an algebraic $k$-equivalence is equivalent to homomorphism indistinguishability over $\calP_k$.  
\end{itemize}

The overall structure of the proof of \cref{thm:main-theorem} is inspired by the proof of the main result of \cite{roberson-seppelt-arxiv}, however, due to the ``noncommutativity" of the NPA hierarchy, additional difficulties arise. For example, the proof of inner-product compatibility of the graph classes $\calL_k$ from~\cite{roberson-seppelt-arxiv} is trivial, while proving the same for our graph classes $\calP_k$ requires a creative construction in \cref{lem:inn-pro-comp}.    

We take a more thorough look at the graph classes $\calP_k$ in \cref{sub-sec:graph-class}. We show that the set of underlying graphs of the union of the bilabelled graph classes $\calP_k$ is the set of all planar graphs. By combining this result with the convergence of the NPA hierarchy, 
we derive a substantially simpler proof of the homomorphism indistinguishability characterization of quantum isomorphism given in \cite{david-laura}. 
In particular, we are able to avoid the use of heavy machinery for dealing with compact quantum groups, which formed one of the main ingredients of the proof in \cite{david-laura}.   

\begin{cor}\label{thm:quant-iso}
	For graphs $G$ and $H$, the following are equivalent:
	\begin{enumerate}
		\item for every $k$, there is a solution for the $k^{\text{th}}$-level of the NPA hierarchy for the $(G,H)$-isomor\-phism game,
		\item $G$ and $H$ are homomorphism indistinguishable over $\bigcup_{k \in \mathbb{N}} \mathcal{P}_k$, the class of all planar graphs,
		\item $G$ and $H$ are quantum isomorphic.
	\end{enumerate}
\end{cor}

The proof of \cref{thm:algorithmic-aspects} relies on the characterisation of the NPA hierarchy as homomorphism indistinguishability relations from \cref{thm:main-theorem}.
That is, instead of attempting to solve the NPA systems of equations, the algorithm decides whether the input graphs are homomorphism indistinguishable over the graph class $\mathcal{P}_k$.
This is done by computing a basis for the finite-dimensional vector space spanned by the homomorphism tensors 
of the bilabelled graphs in $\mathcal{P}_k$. To that end, the algorithm utilises the inductive definition of the graph class $\mathcal{P}_k$ in terms of generators and combinatorial operations. 
Being linear or bilinear, their algebraic counterparts can be used to efficiently compute this basis via a fixed-point procedure, which terminates after polynomially many steps. 
Randomization is only necessary to deal with integers which would otherwise grow to exponential size in the course of the computation.
To overcome this issue, the algorithm relies on linear algebra over finite fields of prime characteristics which are chosen at random.

\subsection{Outline}
The paper begins by covering some of the preliminaries in \cref{sec:preliminaries}. We introduce bilabelled graphs and homomorphism tensors in \cref{sec:hom-ten}, which are our main tools to relate the algebraic question of feasibility of the NPA hierarchy to a combinatorial problem of homomorphism indistinguishability. We then introduce the graph isomorphism game and a suitable version of the NPA hierarchy for the graph isomorphism game in \cref{sec:quant-iso-npa}. The proof of \cref{thm:main-theorem} will be broken down into a series of simpler theorems, namely \cref{thm:main-theorem-1,thm:main-theorem-2,thm:main-theorem-3}. 
This is done in \cref{sec:npa-hom-tensors} and the beginning of \cref{sec:hom-ind}. In \cref{sub-sec:graph-class} we study the graph classes $\calP_k$ in more detail and finish the proof of \cref{thm:quant-iso}, thus providing the promised alternative proof of the main result of~\cite{david-laura}. In \cref{sec:alg-aspects}, we show that there is a polynomial time randomized algorithm for each fixed level of the NPA hierarchy for quantum isomorphism.

\section{Preliminaries}\label{sec:preliminaries}

All graphs in this article are undirected, finite, and without multiple edges, unless stated otherwise. A graph is said to be \emph{simple} if it does not contain any loops. 
A \emph{homomorphism} $h \colon F \to G$ from a graph $F$ to a graph $G$ is a map $V(F) \to V(G)$ such that for all $uv \in E(F)$ it holds that $h(u)h(v) \in E(G)$. Note that this implies that any vertex in $F$ carrying a loop must be mapped to a vertex carrying a loop in $G$.

Write $\hom(F, G)$ for the number of homomorphisms from $F$ to $G$. For a family of graphs $\mathcal{F}$ and graphs $G$ and $H$ we write $G \cong_{\mathcal{F}} H$ if $G$ and $H$ are \emph{homomorphism indistinguishable over $\mathcal{F}$}, i.e., if $\hom(F, G) = \hom(F, H)$ for all $F \in \mathcal{F}$.
Since the graphs $G$ and $H$ into which homomorphisms are counted are throughout assumed to be simple, looped graphs in $\mathcal{F}$ can generally be disregarded as they do not admit any homomorphisms into simple graphs.
For background on homomorphism indistinguishability, see \cite{seppelt_homomorphism_2024}.

Given a finite set $\Omega$, we denote the symmetric group containing all permutations of $\Omega$ by $\mathfrak{S}_{\Omega}$. Given a natural number $k \in \mathbb{N}$, we will denote We use $\mathfrak{S}_{k}$ to denote $\mathfrak{S}_{[k]}$. We shall use $\cyclicpermutations$ to denote the group of cyclic permutations of the set $\Omega \coloneq \{1,\dots,k, 2k, \dots, k+1\}$, i.e., the cyclic subgroup of $\mathfrak{S}_{\Omega}$ generated by the transposition $(12)$. 

\subsection{Bilabelled Graphs and Homomorphism Tensors}\label{sec:hom-ten}

We recall the following definitions from \cite{david-laura,grohe_homomorphism_2022}.
    
For $k, l \geq 1$, a \emph{$(k, l)$-bilabelled graph} is a tuple $\boldsymbol{F} = (F, \boldsymbol{u}, \boldsymbol{v})$ where $F$ is a graph and $\boldsymbol{u} = (u_1, \dots , u_k) \in V(F)^k$, $\boldsymbol{v} = (v_1, \dots, v_l) \in V(F)^l$.
The $\boldsymbol{u}$ are the \emph{in-labelled vertices} of $\boldsymbol{F}$ while the $\boldsymbol{v}$ are the \emph{out-labelled vertices} of $\boldsymbol{F}$. Given a graph $G$, the \emph{homomorphism tensor} of $\boldsymbol{F}$ for $G$ is $\boldsymbol{F}_G \in \mathbb{C}^{V(G)^k \times V(G)^l}$ whose $(\boldsymbol{x}, \boldsymbol{y})$-entry 
is the number of homomorphisms $h \colon F \to G$ such that $h(\boldsymbol{u}_i) = \boldsymbol{x}_i$ and $h(\boldsymbol{v}_j) = \boldsymbol{y}_j$ for all $i \in [k]$ and $j \in [l]$. 

For a $(k, l)$-bilabelled graph $\boldsymbol{F} = (F, \boldsymbol{u}, \boldsymbol{v})$, write $\soe(\boldsymbol{F}) \coloneqq F$ for the underlying unlabelled graph of $\boldsymbol{F}$. 
If $k = l$, write $\tr(\boldsymbol{F})$ for the unlabelled graph underlying the graph obtained from $\boldsymbol{F}$ by identifying $\boldsymbol{u}_i$ with $\boldsymbol{v}_i$ for all $i\in [l]$.
For $\sigma \in \mathfrak{S}_{k+l}$, write $\boldsymbol{F}^\sigma \coloneqq (F, \boldsymbol{x}, \boldsymbol{y})$ where $\boldsymbol{x}_i \coloneqq (\boldsymbol{uv})_{\sigma(i)}$ and  $\boldsymbol{y}_{j-k} \coloneqq (\boldsymbol{uv})_{\sigma(j)}$ for all $1 \leq i \leq k < j \leq k+l$, i.e.\@ $\boldsymbol{F}^\sigma$ is obtained from $\boldsymbol{F}$ by permuting the labels according to $\sigma$.
We also define $\boldsymbol{F}^* \coloneqq (F, \boldsymbol{v}, \boldsymbol{u})$ the graph obtained by swapping in- and out-labels. 
    
Let $\boldsymbol{F} = (F, \boldsymbol{u}, \boldsymbol{v})$ and $\boldsymbol{F}' = (F', \boldsymbol{u}', \boldsymbol{v}')$ be $(k,l)$-bilabelled and $(m,n)$-bilabelled, respectively.
If $l = m$, write $\boldsymbol{F} \cdot \boldsymbol{F}'$ for the $(k, n)$-bilabelled graph obtained from them by \emph{series composition}, whose underlying unlabelled graph obtained from the disjoint union of $F$ and $F'$ by identifying $\boldsymbol{v}_i$ and $\boldsymbol{u}'_i$ for all $i \in [l]$. 
Multiple edges arising in this process are removed. 
The in-labels of $\boldsymbol{F} \cdot \boldsymbol{F}'$ lie on $\boldsymbol{u}$, the out-labels on $\boldsymbol{v}'$.

If $k = m$ and $l = n$ 
write $\boldsymbol{F} \odot \boldsymbol{F}'$ for the \emph{parallel composition} of $\boldsymbol{F}$ and $\boldsymbol{F}'$. 
The underlying unlabelled graph of the $(k, l)$-bilabelled graph $\boldsymbol{F} \odot \boldsymbol{F}'$ is the graph obtained from the disjoint union of $F$ and $F'$ by identifying $\boldsymbol{u}_i$ with $\boldsymbol{u}'_i$ and $\boldsymbol{v}_j$ with $\boldsymbol{v}'_j$  for all $i \in [k]$ and $j \in [l]$. 
Again, multiple edges are dropped. The in-labels of $\boldsymbol{F} \odot \boldsymbol{F}'$ lie on $\boldsymbol{u}$, the out-labels on $\boldsymbol{v}$.
    
As observed in \cite{david-laura,grohe_homomorphism_2022}, the benefit of these combinatorial operations is that they have an algebraic counterpart. Formally, for all graphs $G$ and all $(l, l)$-bilabelled graphs $\boldsymbol{F}, \boldsymbol{F}'$, it holds that $\soe(\boldsymbol{F}_G) = \hom(\soe \boldsymbol{F}, G) $, $\tr(\boldsymbol{F}_G) = \hom(\tr \boldsymbol{F}, G)$, $(\boldsymbol{F}_G)^\sigma = (\boldsymbol{F}^\sigma)_G$, $(\boldsymbol{F} \cdot \boldsymbol{F}')_G = \boldsymbol{F}_G \cdot \boldsymbol{F}'_G$, and $(\boldsymbol{F} \odot \boldsymbol{F}')_G = \boldsymbol{F}_G \odot \boldsymbol{F}'_G$, where $\soe(X)$ denotes the sum of elements, $\tr$ denotes the trace, $\cdot$ denotes matrix multiplication and $\odot$ denotes Schur product.  

Slightly abusing notation, we say that two graphs $G$ and $H$ are homomorphism indistinguishable over a family of bilabelled graphs $\mathcal{S}$, in symbols $G \cong_{\mathcal{S}} H$ if $G$ and $H$ are homomorphism indistinguishable over the family $\{\soe \boldsymbol{S} \mid \boldsymbol{S} \in \mathcal{S}\}$ of the underlying unlabelled graphs of the $\boldsymbol{S} \in \mathcal{S}$.

We conclude this subsection by defining the notion of a minor of bilabelled graphs.

\begin{definition}[\cite{roberson-seppelt-arxiv}]
  Let $\boldsymbol{M}$ and $\boldsymbol{F}$ be $(k,l)$-bilabelled graph. Then, $M$ is said to be a \emph{bilabelled minor} of $\boldsymbol{F}$, written $\boldsymbol{M} \leq \boldsymbol{F}$, if it can be obtained from $\boldsymbol{F}$ by applying a sequence of the following bilabelled minor operations: 
    \begin{enumerate}
        \item edge contraction
        \item edge deletion
        \item deletion of unlabelled vertices.
    \end{enumerate}
\end{definition}

We note that the if a bilabelled graph $\boldsymbol{M}$ is a bilabelled minor of $\boldsymbol{F}$, then $M$ is a minor of $F$ \cite[Lemma 4.12]{roberson-seppelt-arxiv}. Similarly, if $\boldsymbol{F}$ is a bilabelled graph and $M$ is a graph such that $M$ is a minor of $F$, there exists a bilablled graph $\boldsymbol{M'}$ such that $M'$ is the union of $M$ with some unlabelled vertices~\cite[Lemma 4.13]{roberson-seppelt-arxiv}.


\begin{note}[drawing bilabelled graphs]
  Throughout the paper we will depict bilabelled graphs as follows.
  To draw a bilabelled graph $\boldsymbol{F} = (F, \boldsymbol{u}, \boldsymbol{v})$, we draw the graph $F$ and attach $i^{\text{th}}$ input ``wire'', depicted in grey, to $u_i$ and $j^{\text{th}}$ output wire to $v_i$.
  Vertices can have multiple input/output wires attached to them.
  The input and output wires extend to the far right and far left of the picture respectively. 
  Finally, in order to indicate which input/output wire is which, we draw them so that they occur in numerical order (first at the top) at the edges of the picture.
  See Fig.~\ref{fig:bilabelled}(a) for an example of a bilabelled graph and Fig.~\ref{fig:bilabelled}(b) for an illustration of series and parallel composition, defined above.
\end{note}

\begin{figure}
  \centering
  \includegraphics{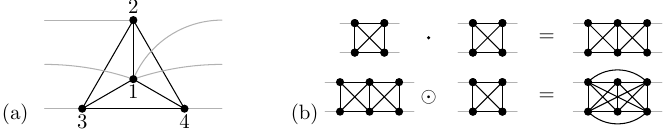} 
  \caption{(a) $\boldsymbol{F} = (K_4, (2,1,3), (1,1,4))$. (b) And example of series and parallel composition.}
  \label{fig:bilabelled}
\end{figure}

\subsection{The Graph Isomorphism Game}\label{sec:quant-iso-npa}

We begin this section by defining the graph isomorphism game. We refer the reader to \cite{ATSERIAS2019289} for more details. 

\begin{definition}
    Let $G$, $H$ be two graphs with $\card{V(G)} = \card{V(H)}$. The $(G,H)$-isomorphism game is cooperative game involving two players Alice and Bob, and the verifier. It is played as follows: 
\begin{itemize}
    \item In each round of the game, the verifier chooses two vertices $g, g' \in V(G)$ (sampled uniformly and independently) and sends them to Alice and Bob respectively. 
    \item Alice and Bob are not allowed to communicate during a round of the game, i.e. after receiving their question from the verifier. However, they are free to strategize before the game starts. 
    \item Alice and Bob respond with vertices $h, h' \in V(H)$, respectively. 
    \item The verifier decides whether Alice and Bob win or lose this round of the game based on the rule function or predicate $V(h,h'\mid g,g')$ which is given by
    \begin{align*}
        V(h,h'\mid g,g') = \begin{cases}
            1 &\quad \text{if } \rel_G(g,g') = \rel_H(h,h') \\
            0 &\quad \text{otherwise}
        \end{cases}
    \end{align*}
\end{itemize}
  Here $\rel_G(g,g') = \rel_H(h, h')$ if and only if both pairs of vertices are adjacent, non-adjacent, or identical.
\end{definition}

Alice and Bob can employ a wide array of strategies to play this game. A \emph{deterministic strategy} is one that involves two functions $f_A, f_B\colon  V(G) \to V(H)$, where Alice and Bob respond with $f_A(g), f_B(g')$ when presented with the questions $g,g'$ respectively. A \emph{perfect strategy} is one that always wins the game for Alice and Bob. 

The predicate necessitates that $f_A = f_B$ for any perfect deterministic strategy $(f_A, f_B)$. Indeed, if $g = g'$, one sees that $f_A(g) = f_B(g)$ for all $g \in V(G)$. Similarly, it is not too difficult to show that $f_A = f_B = f$ is a graph isomorphism, if it is a perfect deterministic strategy. It is also clear that if Alice and Bob answer according to some isomorphism $f\colon V(G) \to V(H)$, then this is a perfect strategy. Hence, the perfect deterministic strategies of the $(G,H)$-isomorphism game are in bijective correspondence with the isomorphisms of $G$ and $H$.

Alice and Bob can also make use of \emph{probabilistic strategies}. A probabilistic strategy is given by joint conditional probability distributions $(p(h,h'\mid g,g'))_{g,g'\in V(G), h,h'\in V(H)}$. Probabilistic strategies are often called \emph{correlations} in the literature. We shall denote the set of correlations indexed by the input sets $X, Y$ and the output sets $A, B$ by $P(X, Y, A, B)$. In  the case where $X = Y$ and $A = B$, we shall use the notation $P(X, A)$ instead. 

It is easy to see that a probabilistic strategy is a perfect strategy for the $(G,H)$-isomorphism game if and only if $V(h,h'\mid g,g') = 0$ implies that $p(h,h'\mid g,g') = 0$ for all $g,g'\in V(G)$ and $h,h' \in V(H)$. We also note that any perfect probabilistic strategy for the $(G,H)$-isomorphism game satisfies that $p(h,h'\mid g,g) = 0$ for all $h \neq h' \in V(H)$ and $g \in V(H)$. Such correlations are known as \emph{synchronous correlations}.

\subsubsection{Quantum Isomorphism of Graphs}

Throughout this paper, we shall be working with what is known as the commuting operator model of quantum mechanics. As discussed earlier, Strategies making use of a shared state are known as \emph{quantum strategies}.
We refer the reader to \cite{Nielsen_Chuang_2010} for a more thorough overview of fundamentals of quantum information.
\begin{definition}
A \emph{quantum strategy} for the $(G,H)$-isomorphism game consists of a shared \emph{state}, i.e. a unit vector $\psi \in \mathcal{H}$ in some Hilbert space $\calH$ and self-adjoint projections $\{E_{g,h}\}_{g\in V(G), h\in V(H)} \subseteq \mathcal{B(H)}$ and $\{F_{g',h'}\}_{g'\in V(G), h'\in V(H)} \subseteq \mathcal{B(H)}$ such that: 
\begin{itemize}
    \item $\sum_{h} E_{gh} = I_{\mathcal{H}}$ and $\sum_{h} F_{g,h} =I_{\mathcal{H}}$ 
    \item $E_{g,h}F_{g',h'} = F_{g',h'}E_{g,h}$ for all $g,g'\in V(G)$ and $h,h'\in V(H)$. 
\end{itemize}

When Alice and Bob receive the questions $g,g'$ from the verifier, they use the PVMs $\{E_{g,h}\}_{h\in V(H)}$ and $\{F_{g',h'}\}_{h'\in V(H)}$ to perform a measurement on their part of the shared state $\psi$. In this case, the conditional probability of outputting $h,h'$ when Alice and Bob receive the questions $g,g'$ is given by $p(h,h'\mid g,g') = \braket{\psi, E_{g,h}F_{g',h'} \psi}$. 
\end{definition}

Two graphs $G$, $H$ are said to be \emph{quantum isomorphic}, written $G \cong_q H$ if there is a perfect quantum strategy for the $(G,H)$-isomorphism game. Two graphs that are isomorphic are also quantum isomorphic as all deterministic 
can be realised as quantum strategies. However, the converse is not true, i.e. there exist non-isomorphic graphs that are quantum isomorphic. Once again, we refer the reader to \cite{ATSERIAS2019289} for further details. 

The existence of a perfect quantum strategy for the $(G,H)$-isomorphism game is characterized by the following proposition:

\begin{prop}[\cite{ATSERIAS2019289}]\label{prop:perf-strat-graph}
    Let $G$, $H$ be two graphs with $\card{V(G)} = \card{V(H)}$. Then, the $G \cong_q H$ if and only if there exist a Hilbert space $\mathcal{H}$ and self-adjoint projections $\{E_{g,h}\}_{g \in V(G), h \in V(H)}$ such that: 
    \begin{enumerate}[label = \roman*.]
        \item $\sum_{h \in V(H)} E_{g,h} = I_{\mathcal{H}}$ for all $g \in V(G)$, 
        \item $\sum_{g \in V(G)} E_{g,h} = I_{\mathcal{H}}$ for all $h \in V(H)$, 
        \item $E_{g,h} E_{g',h'} = 0$ if $V_{G,H}(h,h'\mid g,g') = 0$.
    \end{enumerate}
    Given the first two conditions, the last condition is equivalent to: 
    $$(A(G)\otimes I_{\mathcal{H}})[E_{g,h}]_{g \in V(G), h \in V(H)} = [E_{g,h}]_{g \in V(G), h \in V(H)}(A(H) \otimes I_\mathcal{H}).$$ Furthermore, $G$ and $H$ are isomorphic if and only if there exist mutually commuting self-adjoint projections $\{E_{g,h}\}_{g \in V(G), h \in V(H)}$ satisfying the above conditions.   
\end{prop}

\subsubsection{A Synchronous NPA Hierarchy for Quantum Isomorphism of Graphs}

We shall focus on the NPA hierarchy for the graph isomorphism game in this subsection. We give a more detailed exposition of the NPA hierarchy in \cref{sec:npa}. First, we introduce some notation that will be used throughout the paper and in~\cref{sec:npa}. 

Let $\Sigma = V(G) \times V(H)$. The set of all finite strings in the alphabet $\Sigma$ will be denoted by $\Sigma^*$. Similarly if $k \in \mathbb{N}$, $\Sigma^k$ and $\Sigma^{\leq k}$ denote the set of all strings of length $k$ and the set of all strings of length at most $k$ respectively. $\epsilon$ is used to denote the empty string in $\Sigma^*$, and we use the following operations on strings: 
\begin{itemize}
    \item Let $s \in \Sigma^*$ be a string given by $s_1\cdots s_k$, we denote the \emph{reversed string} $s_k\dots s_1$ by $s^R$. 
    \item Let $s,t \in \Sigma^*$ be strings given by $s_1\cdots s_k$ and $t_1\cdots t_l$ respectively. We denote their \emph{concatenation} $s_1\cdots s_k t_1 \cdots t_l$ by $st$.  
\end{itemize}

With this notation in our hand, we now give the definition of a certificate for the $k^{\text{th}}$-level of the NPA hierarchy for the $(G, H)$-isomorphism game.


\begin{definition}\label{def:cert-npa}
    For two graphs $G$, $H$ with $\card{V(G)} = \card{V(H)}$, a \emph{certificate} for the $k^{\text{th}}$ level of the NPA hierarchy of the $(G,H)$-isomorphism game is a positive semidefinite matrix $\mathcal{R} \in M_{\Sigma^{\leq k}}(\mathbb{C})$ such that: 
    \begin{enumerate}[label = (\roman*)]\label{eq:npa-for-graph-iso}
        \item \label{qiso-item-1}$\mathcal{R}_{\epsilon, \epsilon} = 1$
        \item \label{qiso-item-2} $\mathcal{R}_{s,t} = \mathcal{R}_{s',t'}$ for all $r,s,r',s' \in \Sigma^{\leq k}$, such that $s^Rt \sim (s')^R(t')$, where  we define $\sim$ to be the coarsest equivalence relation satisfying the following two properties:
\begin{itemize}
    \item For each $x,a \in X \times A$, $s(x,a)(x,a)t \sim s(x,a)t$ for all $s,t \in \Sigma^*$.
    \item $st \sim ts$ for all words $s,t \in \Sigma^*$.
\end{itemize} 
        \item \label{qiso-item-3} For all words $s,s' \in \Sigma^{\leq k}$, $g\in V(G)$, $h \in V(H)$ such that $s(g,h)s' \in \Sigma^{\leq k}$, one has 
        \begin{align}
            \sum_{h' \in V(H)} \mathcal{R}_{s(g,h')s', t} = \mathcal{R}_{ss', t} \text{ for all } t\in \Sigma^{\leq k}\label{eq:qiso31}\\
            \sum_{g' \in V(G)} \mathcal{R}_{s(g',h)s', t} = \mathcal{R}_{ss', t} \text{ for all } t\in \Sigma^{\leq k}\\
            \sum_{h' \in V(H)} \mathcal{R}_{t, s(g,h')s'} = \mathcal{R}_{t, ss'} \text{ for all } t\in \Sigma^{\leq k}\\
            \sum_{g' \in V(G)} \mathcal{R}_{t, s(g',h)s'} = \mathcal{R}_{t, ss'} \text{ for all } t\in \Sigma^{\leq k}
        \end{align}
        \item \label{qiso-item-4} For all $s, t \in\Sigma^{\leq k}$, if there exist consecutive $gh, g'h' \in \Sigma$ in $s^Rt$ such that $\rel(g,g') \neq \rel(h,h')$ then $\mathcal{R}_{s, t} = 0$. 
        \end{enumerate}
\end{definition}

We shall say that the $k^{\text{th}}$-level of the NPA hierarchy for the $(G,H)$-isomorphism game is \emph{feasible} if there exists a certificate for the $k^{\text{th}}$-level of the NPA hierarchy for the $(G,H)$-isomorphism game.  

Given a perfect strategy for the $(G,H)$-isomorphism game, where $\{E_{g,h}\}_{g \in V(G), h \in V(H)}$ are Alice's measurement operators and the shared state is $\psi$, we can construct a certificate for any level of the NPA hierarchy in the following way: for level-$k$, we consider the Gram matrix of the vectors $\{E_{g_1,h_1}\cdots E_{g_l, h_l}\psi\}_{g_1h_1\cdots g_lh_l \in \Sigma^{\leq k}}$.

The following proposition shows that the NPA hierarchy for the graph isomorphism game converges, i.e there is a solution for each level of the NPA hierarchy for the $(G,H)$-isomorphism if an only if $G \cong_{q} H$. A proof of the statement can be found in \cref{sec:npa}. 
\Cref{prop:conv-npa-graph} gives two of the implications in \cref{thm:quant-iso}, specifically it shows that items (1) and (3) are equivalent.

\begin{prop}\label{prop:conv-npa-graph}
    Let $G$, $H$ be two graphs with $\card{V(G)} = \card{V(H)}$, and $\Sigma = V(G)\times V(H)$. Then, the $(G,H)$-isomorphism game has a perfect quantum strategy if and only if for each $k \in \mathbb{N}$, there exists a certificate for the $k^{\text{th}}$-level of the NPA hierarchy. 
\end{prop}

\section{NPA Hierarchy and Homomorphism Tensors}\label{sec:npa-hom-tensors}

In this section, we shall give several algebro-combinatorial reformulations of the existence of a certificate for the $k^{\text{th}}$ level of the NPA hierarchy of the graph isomorphism game. These reformulations allow us to interpret the existence of a certificate for the $k^{\text{th}}$ level of the NPA hierarchy as a homomorphism indistinguishability characterization. 
Most of the proofs here are linear-algebraic in nature. 
We shall discuss the graph-theoretic implications of our results in the next section.   

\subsection{Quantum Isomorphism Relaxation via Completely Positive Maps}

In this subsection, we show that a principal submatrix of a certificate for the $k^{\text{th}}$-level of the NPA hierarchy for quantum isomorphism can be interpreted as the Choi matrix (see Appendix~\ref{app:linalg}) of a completely positive map from $M_{V(G)^k}(\mathbb{C})$ to $M_{V(G)^k}(\mathbb{C})$ satisfying certain properties. We then show that the Choi matrix of such a completely positive map uniquely extends to a certificate for the $k^{\text{th}}$-level of the NPA hierarchy for quantum isomorphism. Thus feasibility of the $k^{\text{th}}$-level of the NPA hierarchy for the graph isomorphism game is equivalent to the existence of a completely positive map satisfying certain properties. In order to make these notions precise, we now introduce the atomic bilabelled graphs, which will be the building blocks of the graph classes that we shall construct in the next section.

\begin{definition}\label{def:atomic-q}
  Let $k\geq 1$.
  A $(k, k)$-bilabelled graph $\boldsymbol{F} = (F, \boldsymbol{u}, \boldsymbol{v})$  is atomic if all its vertices are labelled. 
  We define two classes of atomic graphs (see \cref{fig:atomic_graphs}):
  \begin{itemize}
    \item The class $\mathcal{Q}_k^P$ is the class of $(k, k)$-bilabelled minors of the graph
      $\boldsymbol{C}_k \coloneqq (C_k, (1,\dots, k), (k+1,\dots,2k))$ with $V(C_k) = [2k]$ and $E(C_k) = \{\{i,i+1\} : i\in[2k], i\neq k,2k\} \cup \{\{1, k+1\}, \{k, 2k\}\}.$
    \item The class $\mathcal{Q}_k^S$ is the class of $(k,k)$-bilabelled graphs obtained by taking minors of the graph
      $\boldsymbol{M}_k \coloneqq (M_k, (1,\dots, k), (k+1,\dots,2k))$ with $V(M_k) = [2k]$ and $E(M_k) = \{\{i, i+k\} : i\in[k]\}$
  \end{itemize}
  Finally, we define $\mathcal{Q}_k := \mathcal{Q}_k^P \cup \mathcal{Q}_k^S$.

  We also define two specific atomic graphs $\boldsymbol{J}_k, \boldsymbol{I}_k \in \calQ_k$ for each $k \in \mathbb{N}$. 
  \begin{itemize}
      \item $\boldsymbol{J}_k \coloneq (J_k, (1, \dots, k), (k+1, \dots 2k))$ with $V(J_k) = [2k]$ and $E(J_k) = \emptyset$
      \item $\boldsymbol{I}_k \coloneq (I_k, (1, \dots k), (1, \dots k))$ with $V(I_k) = [k]$ and $E(I_k) = \emptyset$
  \end{itemize}
  In other words, $\boldsymbol{J}_k$ and $\boldsymbol{I}_k$ are obtained from $\boldsymbol{M}_k$ by deleting and contracting all the edges respectively. These atomic graphs are special in the sense that one has $(\boldsymbol{J}_k)_G = J$ and $(\boldsymbol{I}_k)_G = I$ for all graphs $G$, where $I$ and $J$ are the identity and all ones matrix respectively. We also note that $\boldsymbol{J}_k \in \calQ_k^S, \calQ_k^P$ and $\boldsymbol{I}_k \in \calQ_k^S$, but $\boldsymbol{I}_k \notin \calQ_k^P$.  
\end{definition}

We can now define quantum isomorphism maps using the homomorphism tensors of the graphs in $\mathcal{Q}_k$:

\begin{definition}\label{def:kqisomap}
    Let $G$ and $H$ be graphs and $k \in \mathbb{N}$. A linear map $\Phi: \mathbb{C}^{V(G)^k \times V(G)^k} \to \mathbb{C}^{V(H)^k \times V(H)^k}$ is a \emph{level $k$ quantum isomorphism map} from $G$ to $H$ if the following holds:
    \begin{align}
        &\quad\Phi \text{ is completely positive,} \label{eq2:q-iso-map-1} \\ 
        &\quad\Phi(\boldsymbol{F}_G \odot X) = \boldsymbol{F}_H \odot \Phi(X) \text{ for all } \boldsymbol{F} \in \mathcal{Q}_k^P \text{ and }X \in \mathbb{C}^{V(G)^k \times V(G)^k},\label{eq2:q-iso-map-2} \\ 
        &\quad\Phi(I) = I = \Phi^*(I), \label{eq2:q-iso-map-3a} \\ 
        &\quad\Phi(J) = J = \Phi^*(J), \label{eq2:q-iso-map-3b} \\ 
        &\quad \Phi(\boldsymbol{F}_G) = \boldsymbol{F}_H \text{ for all } F \in \mathcal{Q}_k,\label{eq2:q-iso-map-4} \\
        &\quad\Phi(\boldsymbol{X}^{\sigma}) = \Phi(\boldsymbol{X})^{\sigma}\text{ for all }\sigma \in \cyclicpermutations. \label{eq2:q-iso-map-5}
    \end{align}
\end{definition}

\begin{figure}[t]
  \centering
  \includegraphics{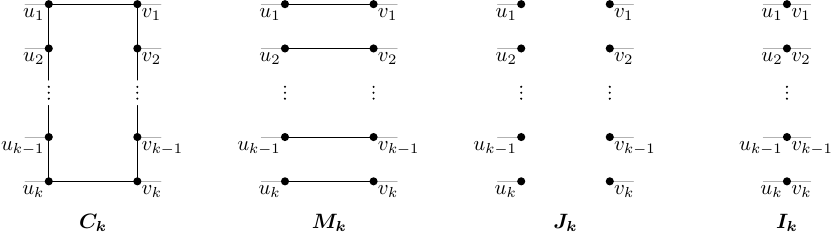}
  \caption{Atomic graphs.}
  \label{fig:atomic_graphs}
\end{figure}
    
\begin{note}\label{rem:cptpunital}
Note that conditions \eqref{eq2:q-iso-map-1} and \eqref{eq2:q-iso-map-3a} state that any level $k$ quantum isomorphism map is a completely positive, trace-preserving, unital map. Also note that condition \eqref{eq2:q-iso-map-4} implies the part of conditions \eqref{eq2:q-iso-map-3a} and \eqref{eq2:q-iso-map-3b} on the map $\Phi$, and thus we are being a bit redundant. However, we do need to explicitly include the conditions on the adjoint $\Phi^*$. Lastly, we are being a bit imprecise since we should really write $\boldsymbol{I}_G, \boldsymbol{I}_H, \boldsymbol{J}_G$, and $\boldsymbol{J}_H$. 
\end{note}

We now prove some lemmas that will be useful for our proof that existence of a level $k$ quantum isomorphism map is equivalent to the existence of a certificate for the $k^\text{th}$ level of the NPA hierarchy.

\begin{lem}\label{lem:Mzeros}
    Let $G$ and $H$ be graphs, $k \in \mathbb{N}$, and suppose that $\Phi$ is a level $k$ quantum isomorphism map from $G$ to $H$. If $M$ is the Choi matrix of $\Phi$, then $M_{s,t} = 0$ if any cyclic permutation of $s^Rt$ contains consecutive terms $gh$ and $g'h'$ such that $\rel(g,g') \ne \rel(h,h')$.
\end{lem}
\begin{proof}
    This follows immediately from condition~\eqref{eq2:q-iso-map-1} and Lemma~\ref{lem:linalg_parallel}.
\end{proof}

\begin{lem}\label{lem:multiplicative}
    Let $G$ and $H$ be graphs and $k \in \mathbb{N}$. If $\Phi$ is a level $k$ quantum isomorphism map from $G$ to $H$, then $\Phi(\boldsymbol{F}_GX) = \boldsymbol{F}_H\Phi(X)$ and $\Phi(X\boldsymbol{F}_G) = \Phi(X)\boldsymbol{F}_H$ for all $\boldsymbol{F} \in \mathcal{Q}_k$ and $X \in \mathbb{C}^{V(G)^k \times V(G)^k}$.
\end{lem}
\begin{proof}
    Since by definition we have that $\Phi(\boldsymbol{F}_G) = \boldsymbol{F}_H$ for all $\boldsymbol{F} \in \mathcal{Q}_k$, by Lemma~\ref{lem:lingalg_series} it suffices to show that for each $\boldsymbol{F} \in \mathcal{Q}_k$, there is some completely positive, trace-preserving, unital map $\Phi'$ such that $\Phi'(\boldsymbol{F}_H) = \boldsymbol{F}_G$. By condition \eqref{eq2:q-iso-map-2} and Lemma~\ref{lem:linalg_parallel}, it follows that $\Phi^*(\boldsymbol{F}_H) = \boldsymbol{F}_G$ for all $F \in \mathcal{Q}_k^P$, so only the $\boldsymbol{F} \in \mathcal{Q}_k^S$ remain.
    
    Let $\hat{\boldsymbol{F}}$ denote the atomic graph which is the minor of the perfect matching bilabeled graph $\mathbf{M}$ obtained by deleting all edges except the one between the first input and first output. Note that this means that $\hat{\boldsymbol{F}}_G = A_G \otimes J^{\otimes k-1}$, where $A_G$ is the adjacency matrix of $G$, and similarly $\hat{\boldsymbol{F}}_H = A_H \otimes J^{\otimes k-1}$. Also note that $\hat{\boldsymbol{F}} \in \mathcal{Q}_k^P \cap \mathcal{Q}_k^S$, and therefore we have that $\Phi^*(\hat{\boldsymbol{F}}_H) = \hat{\boldsymbol{F}}_G$ by the above. Thus by Lemma~\ref{lem:lingalg_series} we have that $\hat{\boldsymbol{F}}_G$ and $\hat{\boldsymbol{F}}_H$ are cospectral, and therefore $A_G$ and $A_H$ are cospectral.

    Now for any $\boldsymbol{F} \in \mathcal{Q}_k^S$, the matrix $\boldsymbol{F}_G$ is a tensor product where each tensor factor is $I$, $J$, or $A_G$, and $\boldsymbol{F}_H$ is the same tensor product except with $A_G$ replaced with $A_H$ everywhere. Since $A_G$ and $A_H$ are cospectral, it follows that all such $\boldsymbol{F}_G$ and $\boldsymbol{F}_H$ are cospectral. Therefore, for any $\boldsymbol{F} \in \mathcal{Q}_k^S$, there is some unitary matrix $U$ such that $U^* \boldsymbol{F}_H U = \boldsymbol{F}_G$. Of course the map $X \mapsto U^* X U$ is completely positive, trace-preserving, and unital and therefore we have completed the proof.
\end{proof}


For our last lemma we need to introduce some notation. For disjoint subsets $R_G, R_H \subseteq [k]$, and $s = g_1h_1\ldots g_kh_k \in (V(G) \times V(H))^k$, we denote by $s(R_G,R_H)$ the set of all strings $s' = g'_1h'_1\ldots g'_kh'_k$ such that $g'_i = g_i$ for $i \notin R_G$ and $h'_i = h_i$ for $i \notin R_H$. For example, if $k = 3$ and $s = g_1h_1g_2h_2g_3h_3$, then 
\[s(\{1\},\{3\}) = \left\{g'_1h_1g_2h_2g_3h'_3 \in (V(G) \times V(H))^3 : g'_1 \in V(G), \ h'_3 \in V(H).\right\}\]
Additionally, for any subset $R \in [k]$ we denote by $s\setminus R$ the substring of $s$ obtained by removing its $i^\text{th}$ entry for each $i \in R$.

\begin{lem}\label{lem:Msums}
    Let $G$ and $H$ be graphs and suppose that $\Phi$ is a level $k$ quantum isomorphism map from $G$ to $H$. If $M$ is the Choi matrix of $\Phi$, then for any $s,t \in (V(G) \times V(H))^k$, disjoint subsets $S_G,S_H \subseteq [k]$, and disjoint subsets $T_G,T_H \subseteq [k]$, we have that
    \begin{equation}\label{eq:Msumlem}
        \sum_{s' \in s(S_G,S_H), \ t' \in t(T_G,T_H)} M_{s',t'}
    \end{equation}
    depends only on the equivalence class of the relation $\sim$ that $(s\setminus (S_G \cup S_H))^R(t\setminus (T_G \cup T_H))$ lies in.
\end{lem}
\begin{proof}
    We first show that the sum only depends $(s\setminus (S_G \cup S_H))^R(t\setminus (T_G \cup T_H))$. We begin by showing that for any fixed $s = g_1h_1\ldots g_kh_k$ and $t = g'_1h'_1\ldots g'_kh'_k$,
    \begin{equation}\label{eq:Msums}
        \sum_{s' \in s(\{l\}, \varnothing)} M_{s',t} = \sum_{s' \in s(\varnothing,\{l\})} M_{s',t}.
    \end{equation}
    To prove this, let $\boldsymbol{M}^{(l)}$ denote the bilabeled graph obtained from the perfect matching bilabeled graph by contracting the edge between the $i^\text{th}$ input and output for $i \ne l$ and deleting the edge between the $l^\text{th}$ input and output. Note that $\boldsymbol{M}^{(l)} \in \mathcal{Q}_k^S$ and both $\boldsymbol{M}^{(l)}_G$ and $\boldsymbol{M}^{(l)}_H$ are equal to $I^{\otimes l-1} \otimes J \otimes I^{\otimes k-l}$. By Lemma~\ref{lem:multiplicative} we have that $\Phi(\boldsymbol{M}^{(l)}_G E^{g_1\ldots g_k, g'_{1}\ldots g'_{k}}) = \boldsymbol{M}^{(l)}_H\Phi(E^{g_1\ldots g_k, g'_{1}\ldots g'_{k}})$, where $E^{g_1\ldots g_k, g'_{1}\ldots g'_{k}}$ is the standard basis matrix with a 1 in the $g_1\ldots g_k, g'_{1}\ldots g'_{k}$-entry and a 0 everywhere else. Equation~\eqref{eq:Msums} is precisely the statement that
\[\left(\Phi(\boldsymbol{M}^{(l)}_GE^{g_1\ldots g_k, g'_{1}\ldots g'_{k}})\right)_{h_1 \ldots h_k,h'_1 \ldots h'_{k}} = \left(\boldsymbol{M}^{(l)}_H\Phi(E^{g_1\ldots g_k, g'_{1}\ldots g'_{k}})\right)_{h_1 \ldots h_k,h'_1 \ldots h'_{k}},\]
and thus it must hold. Similarly,
\begin{equation}\label{eq:Msums2}
    \sum_{t' \in t(\{l\}, \varnothing)} M_{s,t'} = \sum_{t' \in t(\varnothing,\{l\})} M_{s,t'}
\end{equation}
is precisely 
\[\left(\Phi(E^{g_1\ldots g_k, g'_{1}\ldots g'_{k}}\boldsymbol{M}^{(l)}_G)\right)_{h_1 \ldots h_k,h'_1 \ldots h'_{k}} = \left(\Phi(E^{g_1\ldots g_k, g'_{1}\ldots g'_{k}})\boldsymbol{M}^{(l)}_H\right)_{h_1 \ldots h_k,h'_1 \ldots h'_{k}}.\]
Note that this also implies that the lefthand side of Equation~\eqref{eq:Msums} does not depend on the choice of $h_l$, and the righthand side does not depend on the choice of $g_l$. The analogous statement also holds for Equation~\eqref{eq:Msums2}.

Next we will show that the coordinate we are summing over can be ``moved freely". For notational convenience, we will consider the case where we are summing over the first coordinate. We will make use of Lemma~\ref{lem:Mzeros} and the statements proven in the above paragraph. We have that
\begin{align*}
    \sum_{\hat{h}_1 \in V(H)} M_{g_1\hat{h}_1g_2h_2\ldots g_kh_k,g'_1h'_1\ldots g'_kh'_k} &= \sum_{\hat{h}_1 \in V(H)} M_{g_2\hat{h}_1g_2h_2\ldots g_kh_k,g'_1h'_1\ldots g'_kh'_k} \\
    &= M_{g_2h_2g_2h_2\ldots g_kh_k,g'_1h'_1\ldots g'_kh'_k} \\
    &= \sum_{\hat{h}_1 \in V(H)} M_{g_2h_2g_2\hat{h}_1\ldots g_kh_k,g'_1h'_1\ldots g'_kh'_k} \\
    &= \sum_{\hat{h}_1 \in V(H)} M_{g_2h_2g_1\hat{h}_1\ldots g_kh_k,g'_1h'_1\ldots g'_kh'_k}
\end{align*}
Applying the same technique repeatedly allows one to move the $g_1\hat{h}_1$ term to any position in the either of the row or column indices, while keeping the relative cyclic order of the other coordinates fixed. Applying our results so far to several coordinates already shows that the sum~\eqref{eq:Msumlem} from the lemma statement only depends on $(s\setminus (S_G \cup S_H))^R(t\setminus (T_G \cup T_H))$. Moreover, condition~\eqref{eq2:q-iso-map-5} of quantum isomorphism maps implies that $M_{g_{\sigma(1)}h_{\sigma(1)}\dots g_{\sigma(k)}h_{\sigma(k)}, g_{\sigma(k+1)}h_{\sigma(k+1)}\dots g_{\sigma(2k)}h_{\sigma(2k)}}$
is constant for all $\sigma \in \cyclicpermutations$, and from this it follows that for two values of $(s\setminus (S_G \cup S_H))^R(t\setminus (T_G \cup T_H))$ that differ only by a cyclic permutation, the sum~\eqref{eq:Msumlem} is the same.

To complete the proof of the lemma, we only need to show that~\eqref{eq:Msumlem} is the same for two different values of $(s\setminus (S_G \cup S_H))^R(t\setminus (T_G \cup T_H))$ that only differ by replacing two consecutive occurrences of $gh$ by a single occurrence. However, this follows from the fact that in the case where there are two consecutive occurrences, we can replace the constant value of $h$ in one of them by the sum over all possible $h' \in V(H)$.
\end{proof}

\begin{thm}\label{thm:main-theorem-1}
    Let $G$ and $H$ be graphs and $k \in \mathbb{N}$. Then the following are equivalent:
    \begin{enumerate}
        \item The $k^\text{th}$ level of the NPA hierarchy is feasible for the $(G,H)$-isomorphism game.
        \item There exists a level $k$ quantum isomorphism map from $G$ to $H$.
        \item There exists a level $k$ quantum isomorphism map from $H$ to $G$.
    \end{enumerate}
\end{thm}
\begin{proof}
    The idea is to use a principal submatrix of a certificate $\mathcal{R}$ for the $k^\text{th}$ level of the NPA hierarchy as the Choi matrix of a linear map and show that this map is a level $k$ quantum isomorphism map. Conversely, one takes the Choi matrix of a level $k$ quantum isomorphism map and shows that the remaining entries needed in a certificate for the NPA hierarchy can be determined by the entries from the Choi matrix alone. The equivalence of the final item follows from the fact that the first item is clearly symmetric in $G$ and $H$.

    Let $\Sigma = V(G) \times V(H)$ and let $\mathcal{R}$ be a certificate for the $k^\text{th}$ level of the NPA hierarchy for the $(G,H)$-isomorphism game. Define $M \in M_{\Sigma^k}(\mathbb{C})$ entrywise as
    \[M_{s,t} = \mathcal{R}_{s,t} \text{ for all } s,t \in \Sigma_k,\]
    and let $\Phi : \mathbb{C}^{V(G)^k \times V(G)^k} \to \mathbb{C}^{V(H)^k \times V(H)^k}$ be the linear map whose Choi matrix is $M$. Since $\mathcal{R}$ is positive semidefinite so is $M$ and thus $\Phi$ is completely positive, i.e., \eqref{eq2:q-iso-map-1} holds.
    
    Now let $\boldsymbol{F} \in \mathcal{Q}_k^P$ and note that $(\boldsymbol{F}_G)_{g_k\ldots g_1, g_{k+1}\ldots g_{2k}} \in \{0,1\}$ and the value this takes depends only on the values of $\rel(g_i,g_{i+1})$ for $i \in [2k]$ (indices taken modulo $2k$). Similarly, $(\boldsymbol{F}_H)_{h_k\ldots h_1, h_{k+1}\ldots h_{2k}} \in \{0,1\}$ depends (in the same way) only on the values of $\rel(h_i,h_{i+1})$ for $i \in [2k]$. Thus $(\boldsymbol{F}_G)_{g_k\ldots g_1, g_{k+1}\ldots g_{2k}} \ne (\boldsymbol{F}_H)_{h_k\ldots h_1, h_{k+1}\ldots h_{2k}}$ for some $g_1, \ldots, g_{2k} \in V(G)$ and $h_1, \ldots, h_{2k} \in V(H)$, only if there is some $i$ such that $\rel(g_i,g_{i+1}) \ne \rel(h_i,h_{i+1})$. In this case it then follows by ~\cref{def:cert-npa}~\ref{qiso-item-4} that $M_{g_kh_k\ldots g_1h_1,g_{k+1}h_{k+1}\ldots g_{2k}h_{2k}} = 0$. Therefore, by Lemma~\ref{lem:linalg_parallel} we have that \eqref{eq2:q-iso-map-2} holds.

    Next we show that \eqref{eq2:q-iso-map-4} holds. Note that we have already shown that it holds for all $\boldsymbol{F} \in \mathcal{Q}_k^P$. First, consider $\boldsymbol{F} \in \mathcal{Q}_k^S$, i.e., $\boldsymbol{F}$ is a minor of the perfect matching bilabeled graph $\boldsymbol{M}_k$. For each $i \in [k]$ the vertices of $\boldsymbol{F}$ labeled $i$ and $k+i$ are either equal, adjacent, or distinct and non-adjacent. Define $\sim^G_i$ to be the equality relation on $V(G)$ in the first case, the relation of adjacency in the second case, and the trivial relation (i.e., containing all pairs) in the third case. Similarly define $\sim_i^H$, though we will usually drop the superscript. Note that $(\boldsymbol{F}_G)_{g_1\ldots g_k,g_{k+1}\ldots g_{2k}} = 1$ precisely when $g_i \sim_i g_{k+i}$ for all $i \in [k]$, and similarly for $\boldsymbol{F}_H$. We aim to show that $\Phi(F_G) = F_H$. We have that
    \begin{align}
        \Phi(F_G)_{h_1\ldots h_k,h_{k+1}\ldots h_{2k}} &= \sum_{g_1, \ldots, g_{2k}} M_{g_1h_1\ldots g_kh_k,g_{k+1}h_{k+1}\ldots g_{2k}h_{2k}} (F_G)_{g_1\ldots g_k, g_{k+1}\ldots g_{2k}} \\
        &= \sum_{g_1, \ldots, g_{2k}} \mathcal{R}_{g_1h_1\ldots g_kh_k,g_{k+1}h_{k+1}\ldots g_{2k}h_{2k}} (F_G)_{g_1\ldots g_k, g_{k+1}\ldots g_{2k}}\\
        &= \sum_{\substack{g_1, \ldots, g_{2k} \in V(G) \\ \text{s.t. } g_i \sim_i g_{k+i} \ \forall i \in [k]}} \mathcal{R}_{g_1h_1\ldots g_kh_k,g_{k+1}h_{k+1}\ldots g_{2k}h_{2k}}
    \end{align}
Note that every term in the final sum above is 0 unless $h_1 \sim_1 h_{k+1}$, by Conditions \ref{qiso-item-2} and \ref{qiso-item-4} of the certificate $\mathcal{R}$. Assuming, $h_1 \sim_1 h_{k+1}$, we can sum over all $g_1, g_{k+1} \in V(G)$ instead of only those satisfying $g_1 \sim_1 g_{k+1}$. This is because the additional terms we add are all zero. So in this case the above sum is equal to
\begin{equation}
    \sum_{\substack{g_1, \ldots, g_{2k} \in V(G) \\ \text{s.t. } g_i \sim_i g_{k+i} \ \forall i \in [k]\setminus \{1\}}} \mathcal{R}_{g_1h_1\ldots g_kh_k,g_{k+1}h_{k+1}\ldots g_{2k}h_{2k}} = \sum_{\substack{g_2, \ldots, g_{k}, g_{k+2}, \ldots, g_{2k} \in V(G) \\ \text{s.t. } g_i \sim_i g_{k+i} \ \forall i \in [k]\setminus \{1\}}} \mathcal{R}_{g_2h_2\ldots g_kh_k,g_{k+2}h_{k+2}\ldots g_{2k}h_{2k}}
\end{equation}
Now this sum is zero unless $h_2 \sim_2 h_{k+2}$, and in this case we can apply the same argument as above to remove $g_2h_2$ and $g_{k+2}h_{k+2}$ from the sum and index on $\mathcal{R}$. Continuing, we see that $\Phi(\boldsymbol{F}_G)_{h_1 \ldots h_k,h_{k+1}\ldots h_{2k}}$ is zero unless $h_i \sim_i h_{k+i}$ for all $i \in [k]$, and in this case it is equal to $R_{\epsilon,\epsilon} = 1$. In other words, $\Phi(\boldsymbol{F}_G)_{h_1 \ldots h_k,h_{k+1}\ldots h_{2k}} = (\boldsymbol{F}_H)_{h_1 \ldots h_k,h_{k+1}\ldots h_{2k}}$ and so we have shown that $\Phi(\boldsymbol{F}_G) = \boldsymbol{F}_H$ for all $\boldsymbol{F} \in \mathcal{Q}_k^S$. This additionally shows that $\Phi(I) = I$ and $\Phi(J) = J$, since $I$ and $J$ are the homomorphism tensors of $\boldsymbol{I}_k$ and $\boldsymbol{J}_k$ respectively. We can now set $X=J$ in \eqref{eq2:q-iso-map-2} (which we have already shown holds) to obtain that $\Phi(F_G) = F_H$ for all $F \in \mathcal{Q}_k^P$. This completes the proof that \eqref{eq2:q-iso-map-4} holds.

To see that $\Phi^*(I) = I$ and $\Phi^*(J) = J$, note that if $\hat{M}$ is the Choi matrix of the adjoint, then
\[\hat{M}_{h_1g_1\ldots h_kg_k, h_{k+1}g_{k+1}\ldots h_{2k}g_{2k}} = \overline{M_{g_1h_1\ldots g_kh_k, g_{k+1}h_{k+1}\ldots g_{2k}h_{2k}}} = \overline{\mathcal{R}_{g_1h_1\ldots g_kh_k, g_{k+1}h_{k+1}\ldots g_{2k}h_{2k}}}.\]
It is clear from Definition~\ref{def:cert-npa} that the matrix $\hat{\mathcal{R}}$ defined as
\[\hat{\mathcal{R}}_{h_1g_1\ldots h_lg_l, h_{l+1}g_{l+1}\ldots h_{l+r}g_{l+r}}:= \overline{\mathcal{R}_{g_1h_1\ldots g_lh_l, g_{l+1}h_{l+1}\ldots g_{l+r}h_{l+r}}}\]
is a certificate of for the $k^\text{th}$ level of the NPA hierarchy for the $(H,G)$-game, and therefore by the same argument as for $\Phi$, we have that $\Phi^*(I) = I$ and $\Phi^*(J) = J$. Thus \eqref{eq2:q-iso-map-3a} and \eqref{eq2:q-iso-map-3b} hold. The fact that \eqref{eq2:q-iso-map-5} holds is a straightforward consequence of Definition~\ref{def:cert-npa}~\ref{qiso-item-2}. So we have shown that item (1) implies item (2).

Now suppose that (2) holds, i.e., that there is a level $k$ quantum isomorphism map $\Phi$ from $G$ to $H$. Let $M$ be the Choi matrix of $\Phi$. Since $M$ is positive semidefinite, there are vectors $v_{g_1h_1\ldots g_kh_k}$ such that
\begin{equation}
    M_{g_1h_1\ldots g_kh_k,g'_{1}h'_{1}\ldots g'_{k}h'_{k}} = \langle v_{g_1h_1\ldots g_kh_k}, v_{g'_{1}h'_{1}\ldots g'_{k}h'_{k}}\rangle
\end{equation}
Now, for any $0 \le l \le k$ and $g_1, \ldots, g_l \in V(G)$, $h_1, \ldots, h_l \in V(H)$, define
\begin{equation}
    v_{g_1h_1\ldots g_lh_l} = \sum_{h_{l+1},\ldots, h_{k} \in V(H)} v_{g_1h_1\ldots g_kh_k}.
\end{equation}
Of course, we must show that the above definition does not depend on the choice of $g_{l+1}, \ldots, g_{k} \in V(G)$. So consider $g^1_{l+1}, \ldots, g^1_{k}, g^2_{l+1}, \ldots, g^2_{k} \in V(G)$, and let $w$ denote $g_1h_1\ldots g_lh_l$. We have that
\begin{align*}
        & \left \| \sum_{h_{l+1} \dots h_k \in V(H) } v_{g_1h_1 \dots g_l h_l g^1_{l+1}h_{l+1} \dots g^1_k h_k} - \sum_{h_{l+1} \dots h_k \in V(H)} v_{g_1h_1 \dots g_l h_l g^2_{l+1}h_{l+1} \dots g^2_k h_k} \right\| ^2 \\
        & = \sum_{\substack{h_{l+1} \dots h_k \in V(H) \\ h'_{l+1} \dots h'_k \in V(H)}} \left( \begin{aligned} & M_{w g^1_{l+1}h_{l+1} \dots g^1_k h_k, w g^1_{l+1}h'_{l+1} \dots g^1_k h'_k} - M_{w g^1_{l+1}h_{l+1} \dots g^1_k h_k, w g^2_{l+1}h'_{l+1} \dots g^2_k h'_k} \\
        & - M_{w g^2_{l+1}h'_{l+1} \dots g^2_k h'_k, w g^1_{l+1}h_{l+1} \dots g^1_k h_k} + M_{w g^2_{l+1}h_{l+1} \dots g^2_k h_k, w g^2_{l+1}h'_{l+1} \dots g^2_k h'_k} 
        \end{aligned} \right ) \\
        & = \qquad 0
    \end{align*}
by Lemma~\ref{lem:Msums}.

Now define $\mathcal{R} \in M_{\Sigma^{\le k}}(\mathbb{C})$ to be the gram matrix of the vectors $v_{g_1h_1\ldots g_lh_l}$. We will show that $\mathcal{R}$ is a certificate for the $k^\text{th}$ level of the NPA hierarchy for the $(G,H)$-isomorphism game. Clearly $\mathcal{R}$ is positive semidefinite, so it remains to show that it satisfies Conditions~\ref{qiso-item-1}-\ref{qiso-item-4}.

We have that $v_\epsilon = \sum_{h_1,\ldots, h_k \in V(H)} v_{g_1h_1\ldots g_kh_k}$ for any $g_1, \ldots, g_k \in V(G)$ and thus
\begin{align*}
     \mathcal{R}_{\epsilon,\epsilon} &= \langle v_\epsilon, v_\epsilon \rangle \\
     &= \sum_{h_1, \ldots, h_{2k} \in V(H)} M_{g_1h_1\ldots g_kh_k,g_{k+1}h_{k+1}\ldots g_{2k}h_{2k}} \\
     &= \sum_{g_1, \ldots, g_{2k} \in V(H)} M_{g_1h_1\ldots g_kh_k,g_{k+1}h_{k+1}\ldots g_{2k}h_{2k}} \\
     &= \Phi(J)_{h_1\ldots h_k,h_{k+1}\ldots h_{2k}} \\
     &= J_{h_1\ldots h_k,h_{k+1}\ldots h_{2k}} \\
     &= 1.
\end{align*}
Therefore $\mathcal{R}$ satisfies \ref{qiso-item-1}.

We will now show that $\mathcal{R}$ satisfies~\ref{qiso-item-2}, i.e., that $\mathcal{R}_{s,t}$ depends only on the equivalence class of the relation $\sim$ that $s^Rt$ lies in. Let $s = g_1h_1\ldots g_lh_l$ and $t = g'_1h'_1\ldots g'_rh'_r$. We have that

\begin{align}
    \mathcal{R}_{s,t} &= \langle v_s, v_t \rangle \\
    &= \sum_{h_{l+1}, \ldots, h_k, h'_{r+1} \ldots, h'_k \in V(H)} \langle v_{g_1h_1\ldots g_kh_k}, v_{g'_1h'_1\ldots g'_kh'_k} \rangle \\
    &= \sum_{h_{l+1}, \ldots, h_k, h'_{r+1} \ldots, h'_k \in V(H)} M_{g_1h_1\ldots g_kh_k, g'_1h'_1\ldots g'_kh'_k} \label{eq:RfromM}
\end{align}
which only depends on the equivalence class of $s^Rt$ by Lemma~\ref{lem:Msums}.

Next we will show that $\mathcal{R}$ satisfies \ref{qiso-item-3}. Recall that Equation~\eqref{eq:qiso31} from \ref{qiso-item-3} states that
\[\sum_{h' \in V(H)} \mathcal{R}_{s(g,h')s',t} = \mathcal{R}_{ss',t}\]
for $s,s',t \in \Sigma^{\le k}$ such that $s(g,h')s' \in \Sigma^{\le k}$. But this of course follows from Lemma~\ref{lem:Msums} and the expression for $\mathcal{R}_{s,t}$ given in \eqref{eq:RfromM}. 

Similarly, \ref{qiso-item-4} holds by Lemma~\ref{lem:Msums}, Lemma~\ref{lem:Mzeros}, and Equation~\eqref{eq:RfromM}. Therefore $\mathcal{R}$ is a certificate for the $k^\text{th}$ level of the NPA hierarchy for the $(G,H)$-isomorphism game and we are done.
\end{proof}

We remark that it follows from Definition~\ref{def:cert-npa}~\ref{qiso-item-2} that the extension of the Choi matrix $M$ to the certificate $\mathcal{R}$ is in fact unique.

The observant reader may have noticed that not all elements of $\mathcal{Q}_k^S$ were needed to prove that item (2) implies item (1) above. This is related to the fact that it is in fact possible to generate this unused bilabelled graphs from those that were used in the proof. However, redefining the set $\mathcal{Q}_k^S$ to only contain those that were used does not make it any easier to prove that (1) implies (2), and we will need these extra graphs later.

\subsection{Isomorphisms Between Matrix Algebras}

 In this subsection, we shall see how a quantum isomorphism map restricts to a homomorphism between algebras containing homomorphism tensors for $G$ to homomorphism tensors for $H$ of graphs in $\calQ_k$. This brings us a step closer to interpreting a solution for the $k^{\text{th}}$-level of the NPA hierarchy as a homomorphism indistinguishability result. 

A matrix algebra $\calA \subseteq M_{n^k}(\complex)$ is \emph{$S$-partially coherent} if it is unital, self-adjoint, contains $J$, and is closed under Schur product with any matrix in $S$. 
Further, $\calA$ is \emph{cyclically-symmetric} if $A^\sigma \in \calA$, for every $A \in \calA$ and $\sigma \in \cyclicpermutations$.

\begin{definition}
  Let $S_k$ be the set of homomorphism tensors of $(k, k)$-bilabelled atomic graphs for $G$ in $\calQ_k^P$.
  For a graph $G$, we define the algebra $\widehat{\mathcal{Q}}_G^k$ as the minimal cyclically-symmetrical $S_k$-partially coherent algebra containing homomorphism tensors of all $(k,k)$-bilabelled graphs in $\calQ_k$ for $G$.
\end{definition}

\begin{definition}\label{def:algebraic-k-quiv}
  Two $n$-vertex graphs $G$ and $H$ are \emph{algebraically $k$-equivalent} if there is \emph{algebraic $k$-equivalence}, i.e., a vector space isomorphism $\varphi\colon \widehat{\mathcal{Q}}_G^k \to \widehat{\mathcal{Q}}_H^k$ such that
  \begin{enumerate}
    \item $\varphi(M^*) = \varphi(M)^*$ for all $M\in\widehat{\mathcal{Q}}_G^k$,
    \item $\varphi(MN) = \varphi(M)\varphi(N)$ for all $M,N\in\widehat{\mathcal{Q}}_G^k$,
n    \item $\varphi(\boldsymbol{F}_G \odot M) = \boldsymbol{F}_H\odot\varphi(M)$ for all 
    $\boldsymbol{F}\in\mathcal{Q}_k^P$ and any $M\in\widehat{\mathcal{Q}}_G^k$,
    \item $\varphi(I) = I$, $\varphi(J) = J$  and $\varphi(\boldsymbol{F}_G) = \boldsymbol{F}_H$ for all $\boldsymbol{F}\in\mathcal{Q}_k$,
    \item $\varphi(M^\sigma) = \varphi(M)^\sigma$ for all $M\in\widehat{\mathcal{Q}}_G^k$ and $\sigma \in \cyclicpermutations$.
    \item $\varphi$ is trace preserving. 
  \end{enumerate}
\end{definition}

Note that every algebraic $k$-equivalence is sum-preserving, i.e., $\soe(\varphi(X)) = \soe(X)$ for all $X \in \hat{\calQ}_k$. Indeed, $\soe(\varphi(X)) = \tr(J\varphi(X)) = \tr(\varphi(JX)) = \tr(JX) = \soe(X)$. 

\begin{thm}\label{thm:main-theorem-2}
  Let $k\geq 1$. 
  Two graphs $G$ and $H$ are algebraically $k$-equivalent if and only if there is a level-$k$ quantum isomorphism map from $G$ to $H$.
\end{thm}
\begin{proof}
  Let $\Phi\colon M_{V(G)^k}(\complex) \to M_{V(H)^k}(\complex)$ be a level-$k$ quantum isomorphism map from $G$ to $H$.
  By~\eqref{eq2:q-iso-map-4}, we have $\Phi(\boldsymbol{F}_G) = \boldsymbol{F}_H$ for all $\boldsymbol{F}\in \calQ_k$. 
  Similarly $\Phi^*(\boldsymbol{F}_H) = \boldsymbol{F}_G$, for all $\boldsymbol{F}\in \calQ_k$.
  In particular, $\Phi$ and $\Phi^*$ are unital and also trace-preserving. 
  By \cref{lem:multiplicative}, for any $W \in M_{V(G)^k}(\mathbb{C})$ we have that $\Phi(\boldsymbol{F}_GW) = \boldsymbol{F}_H\Phi(W)$ and $\Phi(W\boldsymbol{F}_G) = \Phi(W)\boldsymbol{F}_H$ for all atomic $\boldsymbol{F} \in Q_k$.
  By~\eqref{eq2:q-iso-map-2} we have that $\Phi(\boldsymbol{F}_G\odot X)$ = $\boldsymbol{F}_H\odot\Phi(X)$ for all $\boldsymbol{F} \in \calQ_k^P$ and $X \in \complex^{V(G)^k\times V(G)^k}$. Hence the restriction of $\Phi$ to $\widehat{\mathcal{Q}}_G^k$ gives us an algebraic $k$-equivalence from $G$ to $H$. 
  
      Conversely, suppose that $\varphi \colon \widehat{\calQ}^k_G \to\widehat{\calQ}^k_H$ is an algebraic $k$-equivalence. Now, by \cref{lem:linalg_unitary}, there exists a unitary matrix $U \in M_{{V(G)}^k}(\mathbb{C})$ such that $\varphi(X) = UXU^*$ for all $X \in \widehat{\calQ}^k_G$.
		Let $\widehat{\varphi} \colon M_{V(G)^k}(\complex) \to M_{V(H)^k}(\complex)$ be the map given by $\widehat{\varphi}(X) = UXU^*$.
		Let $\Pi \colon M_{V(G)^k}(\complex) \to \widehat{\mathcal{Q}}^k_G$ be the orthogonal projection onto $\widehat{\mathcal{Q}}^k_G$.
		Define $\Phi \colon M_{V(G)^k}(\complex) \to M_{V(H)^k}(\complex)$ by $\Phi \coloneqq \widehat{\varphi} \circ \Pi$.
		We know that $\widehat{\varphi}$ is completely positive, trace-preserving, and unital; by \cite[Lemma~5.4]{david_mathprog}, so is $\Pi$ and hence their composition $\Phi$ is also a unital completely positive trace-preserving map. 
		
		Furthermore, $\Pi(J) = J$, which implies $\Phi(J) =J = \Phi^*(J)$. Consider the linear map $\Lambda_\sigma \colon X \mapsto X^\sigma$ for $\sigma \in \cyclicpermutations$. Since $\widehat{\mathcal{Q}}^k_G$ is closed under the action of $\cyclicpermutations$, it holds that $\Lambda_\sigma \circ \Pi = \Pi \circ \Lambda_\sigma \circ \Pi $. Furthermore, $(\Lambda_{\sigma})^{*} = \Lambda_{\sigma^{-1}}$ and $\Pi$ is self-adjoint, i.e.\@ $\Pi^* = \Pi$. Hence,
		\[
		\Pi \circ \Lambda_{\sigma}
		= \Pi^* \circ \Lambda_\sigma
		= (\Lambda_{\sigma^{-1}} \circ \Pi)^*
		= (\Pi \circ \Lambda_{\sigma^{-1}} \circ \Pi)^*
		= \Pi \circ \Lambda_{\sigma} \circ \Pi
		= \Lambda_\sigma \circ \Pi.
		\]
		So $\Pi$ and $\Lambda_\sigma$ commute. Hence,
		\begin{align*}
		\Phi(X^\sigma) & = (\widehat{\varphi} \circ \Pi \circ \Lambda_\sigma)(X) \\
			& = (\widehat{\varphi} \circ \Lambda_\sigma \circ \Pi )(X) \\
            & = \widehat{\varphi} ((\Pi(X))^{\sigma}) \\
            & =  \widehat{\varphi}(\Pi(X))^{\sigma} \\
			& = \Phi(X)^\sigma
		\end{align*}
    For $\boldsymbol{F} \in \calQ_k^P$, the map $\Lambda_{\boldsymbol{F}_G}(X) = \boldsymbol{F}_G \odot X$ is an orthogonal projection satisfying $\Lambda_{\boldsymbol{F}_G} \circ \Pi = \Pi \circ \Lambda_{\boldsymbol{F}_G} \circ \Pi$, and so we can similarly show that $\Phi(\boldsymbol{F}_G \odot X) = \boldsymbol{F}_H \odot \Phi(X)$ for all $X \in M_{V(G)^k}(\complex)$. We direct the reader to the proof of a similar statement \cite[Theorem 6.2]{david_mathprog} for more details. 
\end{proof}

\section{Homomorphism Indistinguishability}\label{sec:hom-ind}

In this section, we shall finish the proof of the main theorem by constructing the graph classes $\calP_k$ such that homomorphism indistinguishability over $\mathcal{P}_k$ is equivalent to the feasibility of the $k^{\text{th}}$-level of the NPA hierarchy. As stated earlier, we shall start with the atomic graphs $\calQ_k$ as our building blocks and construct the graph class $\calP_k$ by series composition, cyclic permutation of labels, and parallel composition with appropriate graphs. 

\begin{definition}\label{def:pk}
  Let $\mathcal{P}_k$ be the class of $(k, k)$-bilabelled graphs generated by the set of atomic graphs $\mathcal{Q}_k$ under parallel composition with graphs from $\mathcal{Q}_k^P$, series composition, and the action of the group $\cyclicpermutations$ on the labels.
\end{definition}

We remark that the action of $\cyclicpermutations$ on a bilabelled graph $\boldsymbol{F} \in \calP_k$ corresponds to ``rotating" the drawing of $\mathcal{F}$ (see e.g.~Figure~\ref{fig:inprod_compatibilty}).

\subsection{Inner-product Compatibility of $\calP_k$}\label{sub-sec:inprodcomp}

A class of $(k,k)$-bilabelled graphs $\calT$ is said to be \emph{inner-product compatible} if for all $\boldsymbol{R}, \boldsymbol{S} \in \calT$, there is a $\boldsymbol{Q} \in \calT$ such that $\tr(\boldsymbol{R}^*\cdot \boldsymbol{S}) = \soe(\boldsymbol{Q})$. 

\begin{lem}\label{lem:inn-pro-comp}
    The graph classes $\calP_k$ are inner-product compatible for each $k \in \mathbb{N}$. 
\end{lem}
\begin{proof}
    We first deal with the even case. Let $\boldsymbol{T}$ and $\boldsymbol{T^*}$ be the two graphs obtained by rotating the atomic graph $\boldsymbol{F} \in \calP_{2k}$ as in \cref{fig:inprod_compatibilty}(a). Let $\boldsymbol{R}, \boldsymbol{S} \in \calP_{2k}$ be two bilabelled graphs. Let $(\boldsymbol{R}^*\boldsymbol{S})^{90^\circ}$ be the graph obtained by rotating $\boldsymbol{R}^*\boldsymbol{S}$ as in \cref{fig:inprod_compatibilty}(b). We note that $\boldsymbol{R}^*\boldsymbol{S}, (\boldsymbol{R}^*\boldsymbol{S})^{90^\circ} \in \calP_{2k}$. It is now apparent from \cref{fig:inprod_compatibilty}(c) that $\tr(\boldsymbol{R}^*\boldsymbol{S}) = \soe(\boldsymbol{T} \cdot (\boldsymbol{R}^*\boldsymbol{S})^{90^\circ} \cdot \boldsymbol{T^*})$, and evidently we have that $(\boldsymbol{T} \cdot (\boldsymbol{R}^*\boldsymbol{S})^{90^\circ} \cdot \boldsymbol{T^*}) \in \calP_{2k}$. Since $\boldsymbol{R}$ and $\boldsymbol{S}$ were arbitrary, this shows that $\calP_{2k}$ is inner-product compatible for all $k \in \mathbb{N}$. 

    Similarly, for the odd case consider graphs $\boldsymbol{R}, \boldsymbol{S} \in \calP_{2k+1}$. Let $\boldsymbol{T}$ and $\boldsymbol{T^*}$ be the graphs obtained by rotating the atomic graph $\boldsymbol{F}$ as in \cref{fig:inprod_compatibilty_odd}(a), and $(\boldsymbol{R}^*\boldsymbol{S})^{90^\circ}$ be the graph obtained by rotating $\boldsymbol{R}^*\boldsymbol{S}$ as in \cref{fig:inprod_compatibilty_odd}(b). We can now see from \cref{fig:inprod_compatibilty_odd}(c) that $\tr(\boldsymbol{R}^*\boldsymbol{S}) = \soe((\boldsymbol{T} \cdot (\boldsymbol{R}^*\boldsymbol{S})^{90^\circ} \cdot \boldsymbol{T^*}) \odot \boldsymbol{O}_k)$, and it is obvious that $(\boldsymbol{T} \cdot (\boldsymbol{R}^*\boldsymbol{S})^{90^\circ} \cdot \boldsymbol{T^*}) \odot \boldsymbol{O}_k \in \calP_{2k+1}$, where $\boldsymbol{O}_k$ is the $(k,k)$-bilabelled graph with a single vertex and no edges. This establishes that $\calP_{2k+1}$ is inner-product compatible for all $k \in \mathbb{N}$, which finishes the proof. 
\end{proof}

\begin{figure}[t]
  \centering
  \includegraphics{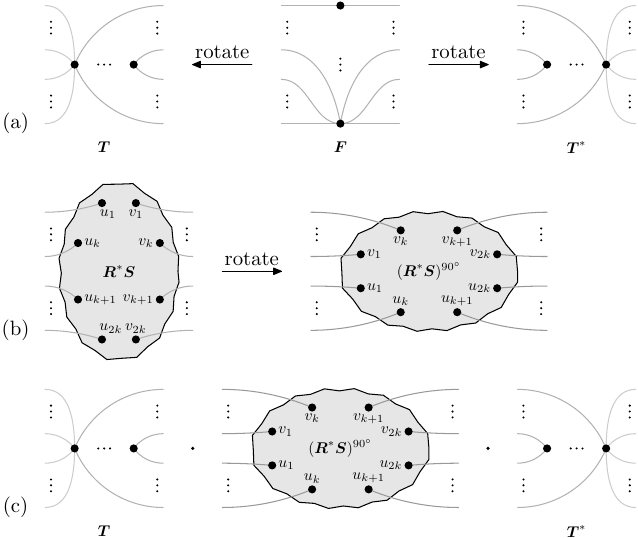}
  \caption{Inner-product compatibility in the even case.}
  \label{fig:inprod_compatibilty}
\end{figure}

\begin{thm}[{\cite[Theorem 4.6]{roberson-seppelt-arxiv}}]\label{thm:inn-prod-comp}
    Let $\calS$ be an inner-product compatible class of $(k,k)$-bilabelled graphs containing $J$. Write $\calS_G$ for the homomorphism tensors $\{\boldsymbol{F}_G \mid \boldsymbol{F} \in \mathcal{S}\}$, and let $\mathbb{C}S_G \subseteq M_{V(G)^k}(\mathbb{C})$ denote the vector space spanned by $\calS_G$. Then, the following are equivalent:
    \begin{enumerate}[label = (\roman*)]
        \item $G$ and $H$ are homomorphism indistinguishable over $\calS$. 
        \item there exists a sum-preserving vector space isomorphism $\varphi: \mathbb{C}\calS_G \to \mathbb{C}\calS_H$ such that $\varphi(\boldsymbol{F}_G) = \boldsymbol{F}_H$ for all $\boldsymbol{F} \in S$. 
    \end{enumerate}
\end{thm}

\begin{figure}[t]
  \centering
  \includegraphics{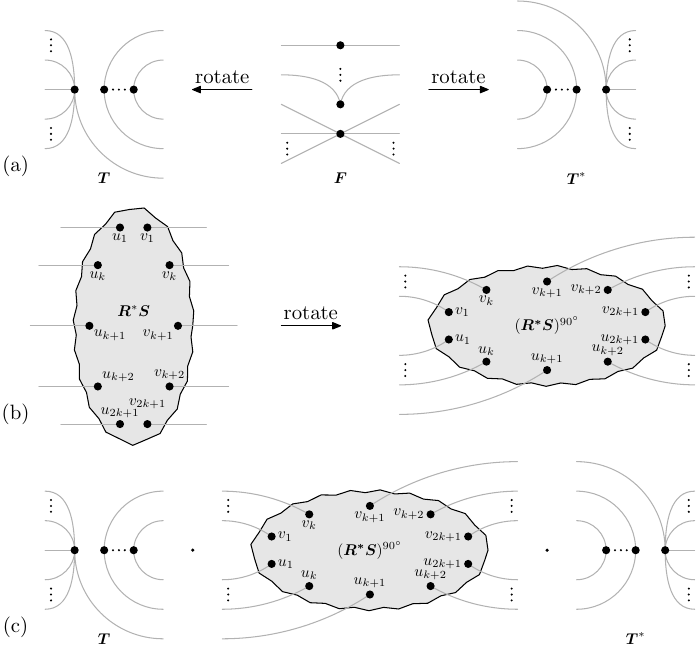}
  \caption{Inner-product compatibility in the odd case.}
  \label{fig:inprod_compatibilty_odd}
\end{figure}

The next theorem completes the proof of \cref{thm:main-theorem}. 

\begin{thm}\label{thm:main-theorem-3}
  Let $k\geq 1$. 
  Two graphs $G$ and $H$ are homomorphism indistinguishable over $\calP_k$ if and only if they are partially $k$-equivalent.
\end{thm}
\begin{proof}
The main idea of the proof is to note that $\mathbb{C}(\calP_k)_G = \hat{\calQ}_k^G$ for any graph $G$. Now, by \cref{thm:inn-prod-comp}, it is clear that $G$ and $H$ are homomorphism indistinguishable over $\calP_k$ if and only if there is a sum-preserving vector space isomorphism from $ \hat{\calQ}_k^G$ to $ \hat{\calQ}_k^H$. Since every algebraic $k$-equivalence is sum-preserving, it is clear that if $G$ and $H$ are algebraically $k$-equivalent, then $G \cong_{\calP_k} H$. 

On the contrary, if we assume that $G \cong_{\calP_k} H$, then there is a sum preserving vector space isomorphism $\varphi\colon \hat{\calQ}_k^G$ to $ \hat{\calQ}_k^H$ such that $\varphi(\boldsymbol{F}_G) = \boldsymbol{F}_H$ for all $\boldsymbol{F} \in \calP_k$. Hence, if $\boldsymbol{F}_1, \boldsymbol{F}_2 \in \calP_k$, then
\[\varphi((\boldsymbol{F}_1)_G(\boldsymbol{F}_2)_G) = \varphi((\boldsymbol{F}_1 \cdot \boldsymbol{F}_2)_G) = (\boldsymbol{F}_1 \cdot \boldsymbol{F}_2)_H = (\boldsymbol{F}_1)_H(\boldsymbol{F}_2)_H = \varphi((\boldsymbol{F}_1)_G) \varphi((\boldsymbol{F}_2)_G).\]
Similarly, if $\boldsymbol{F}_2 \in \calQ_k^P$, then we can show that $\varphi((\boldsymbol{F}_1)_G \odot (\boldsymbol{F}_2)_G) = \varphi((\boldsymbol{F}_1)_G) \odot \varphi((\boldsymbol{F}_2)_G)$. Other than $\varphi$ being trace-preserving, the remaining conditions from \cref{def:algebraic-k-quiv} can be proven in similar manners. The fact $\varphi$ is trace-preserving follows from it being sum-preserving and $\calP_k$ being inner-product compatible.

\end{proof}

\subsection{Planarity and Minor-closedness of $\calP_k$}\label{sub-sec:graph-class}

In this subsection, we shall look at the graph classes $\calP_k$ in more detail. The first thing we note is that for each $k \in \mathbb{N}$, one has that $\calP_k \subseteq \calL_k$, which were the classes of graphs constructed in \cite{roberson-seppelt-arxiv} to obtain a homomorphism indistinguishability charatcterization for the Lasserre hierarchy of SDP relaxations of graph isomorphism. We include a brief overview of this in \cref{sec:lasserre} and we refer the reader to \cite{roberson-seppelt-arxiv} for more details. 
Indeed, for each $k \in \mathbb{N}$ it is not too difficult to see that $\calQ_k \subseteq \calA_k$ (where $\calA_k$ are the atomic graphs used to construct $\calL_k$ in \cite{roberson-seppelt-arxiv}). Moreover, the set of allowed operations while constructing $\calP_k$ from $\calQ_k$ is also a subset of the set of operations allowed while constructing $\calL_k$ from $\calA_k$. In fact, for $k = 1$, it is not too difficult to show that $\calP_k = \calL_k$ is the set of all outerplanar graphs. We note here that the class $\calP_k$ is not the set of $k$-outerplanar graphs for $k \neq 1$. For any $k \in \mathbb{N}$, one can construct a $k$-outerplanar graph in $\calP_2$. The following lemma is immediate from \cite[Lemma~4.7]{roberson-seppelt-arxiv}:

\begin{lem}\label{lem:bounded-tw}
    Let $F \in \calP_k$ be a $(k,k)$-bilabelled graph. Then, the treewidth of the underlying graph $\soe F$ is at most $3k -1$. 
\end{lem}

We now work towards giving an alternate proof of the main result of \cite{david-laura} that shows that two graphs are quantum isomorphic if and only if they are homomorphism indistinguishable over all planar graphs. 
Let $\mathcal{P} \coloneqq \bigcup_{k=1}^k \calP_k$. Then, it follows from the main theorem (\cref{thm:main-theorem}) and convergence of NPA hierarchy (\cref{prop:conv-npa-graph}) that $G \cong_q H$ is equivalent to $G \cong_{\calP} H$. This proves \cref{thm:quant-iso} with the exception of the claim in item (2) that $\calP$ is the set of all planar graphs. We now work towards proving this final claim, thus completing the proof of \cref{thm:quant-iso}.  

First, we show that for each $k \in \mathbb{N}$, the graph class $\calP_k$ only contains planar graphs. We begin by precisely defining what it means for a bilabelled graph to be planar. We use the definition given in \cite{david-laura}. 

\begin{definition}
    Given a $(k,l)$-bilabelled graph $\boldsymbol{G} = (G, \boldsymbol{a}, \boldsymbol{b})$ we define the graph $G^0$ as the graph obtained from $G$ 
    by adding a cycle $C = \alpha_1, \dots, \alpha_k, \beta_l, \dots , \beta_1$ 
    and the edges $a_i\alpha_i$ and $b_j\beta_j$ for each $i \in [k]$ and $j \in [l]$, and say that $\boldsymbol{G}$ is planar if $G^{0}$ has a planar embedding where the cycle $C$ is the boundary of a face. We shall refer to $C$ as the \emph{enveloping cycle} of $G^0$, and usually consider planar embeddings where $C$ is the boundary of the outer face.
\end{definition}

\begin{lem}\label{lem:planar-class}
    For each $k \in \mathbb{N}$, the class of graphs $\calP_k$ is contained in the set of all planar bilabelled graphs. 
\end{lem}
\begin{proof}
    We prove this by induction. It is straightforward to see that all elements of $\calQ_k$ are planar bilabelled graphs. It follows from \cite[Lemma 5.12]{david-laura} that 
    if $\boldsymbol{G}$, $\boldsymbol{H}$ are $(k,k)$-bilabelled planar graphs, then so is $\boldsymbol{G} \cdot \boldsymbol{H}$. It is also clear that if $\boldsymbol{F}$ is planar, then so is $\boldsymbol{F}^\sigma$ for any $\sigma \in \cyclicpermutations$. Thus is only remains to show that if $\boldsymbol{F}$ is a planar bilabeled graph, then so is $\boldsymbol{F} \odot \boldsymbol{F}'$ for any $\boldsymbol{F}' \in \calQ_k^P$. We will show that $\boldsymbol{F} \odot \boldsymbol{F}'$ can be constructed from $\boldsymbol{F}$ and other planar bilabeled graphs through series composition and cyclic permutations. It then follows by the above that $\boldsymbol{F} \odot \boldsymbol{F}'$ is planar.

    Let Let $\boldsymbol{F} = (F,(u_1,\ldots, u_k),(v_1,\ldots,v_k))$ and $\boldsymbol{F}' = (F',(u'_1,\ldots, u'_k),(v'_1,\ldots,v'_k))$. Since $\boldsymbol{F}'$ is a minor of the cycle $\boldsymbol{C}_k$, the bilabelled graph $\boldsymbol{F} \odot \boldsymbol{F}'$ is obtained from $\boldsymbol{F}$ by adding an edge between $u_i$ and $u_{i+1}$ when $u'_iu'_{i+1} \in E(F')$, and similarly for $v_{i}$ and $v_{i+1}$, $u_1$ and $v_1$, and $u_k$ and $v_k$. Define $\boldsymbol{U} = (U,(u_1,\ldots,u_k),(u'_1,\ldots,u'_k))$ where $V(U) = \{u'_i \in V(F') : i \in [k]\}$ and $u'_iu'_{i+1} \in E(U)$ if $u'_iu'_{i+1} \in E(F')$. Define $\boldsymbol{V}$ similarly. Then the bilabelled graph $\boldsymbol{U} \cdot \boldsymbol{F} \cdot \boldsymbol{V}$ is almost $\boldsymbol{F} \odot \boldsymbol{F}'$, except that we may need to identify or make adjacent the vertices $u_1$ and $v_1$, and similarly for $u_k$ and $v_k$. This can easily be achieved by applying a cyclic permutation to $\boldsymbol{U} \cdot \boldsymbol{F} \cdot \boldsymbol{V}$ and then performing series composition on the left and right by appropriate bilabelled graphs similar to $\boldsymbol{U}$ and $\boldsymbol{V}$. It is easy to see that $\boldsymbol{U}$ and $\boldsymbol{V}$ are planar and thus we have show that $\boldsymbol{F} \odot \boldsymbol{F}'$ must be.

\end{proof}

We will also need that the class $\calP_k$ is closed under taking minors (recall the notion of minors of bilabelled graphs from~\cref{sec:hom-ten}) and that $\soe(\calP_k)$ is closed under minors and disjoint unions. The proofs of these facts are almost identical to the analogous proofs for the class $\mathcal{L}_k$ used in~\cite{roberson-seppelt-arxiv}. Therefore, we will just state the results and comment on the very minor differences in the proofs.




\begin{lem}\label{lem:minor-closedbl}
    For each $k \in \mathbb{N}$, the class of bilabelled graphs $\calP_k$ is minor-closed. 
\end{lem}
The proof of this follows the proof of Lemma~4.16 of~\cite{roberson-seppelt-arxiv} almost exactly. The only difference is as follows. In the proof, at one point some care must be taken to remove some unwanted isolated labelled vertices that become isolated and unlabeled after a series composition. To do this, prior to the series composition, a parallel composition is used which has the effect of merging the $i^\text{th}$ labelled output vertex with the first input vertex. Such a parallel composition is not available in $\calP_k$, rather only cyclically consecutive labelled vertices can be merged. However, the choice of merging with the first input was arbitrary, and the only thing truly necessary for the proof is that the unwanted isolated labelled vertex is merged with some labelled vertex which does not become an unwanted isolated unlabelled vertex after the series composition. This is easily accomplished in $\calP_k$, since one can merge the unwanted isolated labeled vertex with one that is next to it in cyclic order (if this is also an unwanted isolated labelled vertex, then merge both with the next, etc.).

With the above lemma in hand, the proof of the following is identical to Lemma~4.9 in~\cite{roberson-seppelt-arxiv}.

\begin{lem}\label{lem:minor-closed}
    For each $k \in \mathbb{N}$, the class of graphs $\soe(\calP_k)$ is minor-closed and union-closed. 
\end{lem}

We now work towards showing that although each $\calP_k$ has bounded treewidth, their union contains arbitrarily large grids  

\begin{lem}\label{lem:unbounded-tw}
     For each $k \in \mathbb{N}$, the class of graphs $\calP_k$ contains the $k \times k$ grid.   
\end{lem}
\begin{proof}
  Let $\boldsymbol{V}_k \in \calP_k$ denote the graph obtained by the parallel composition of $\boldsymbol{C}_k$ and $\boldsymbol{M}_k$, i.e $\boldsymbol{V}_k \coloneqq \boldsymbol{M}_k \odot \boldsymbol{C}_k$. Note that this parallel composition is allowed as $\boldsymbol{C}_k \in \calQ_k^P$. Clearly, $\soe(\boldsymbol{V}_k)$ is the vertical grid with $2$ columns and $k$ rows (see Figure~\ref{fig:grid}). Define $\boldsymbol{G}_k$ to be the graph constructed by the series composition of $k-1$ copies of $\boldsymbol{V}_k$, i.e. $\boldsymbol{G}_k = \boldsymbol{V}_k^{(1)} \cdot \ \dots \ \cdot V_k^{(k-1)}$, 
  where each $\boldsymbol{V}_k^{(i)}$ is an isomorphic copy of $\boldsymbol{V}_k$. It is clear from Figure~\ref{fig:grid} that $\soe \boldsymbol{G}_k$ is the $k \times k$ grid. 
\end{proof}

\begin{figure}[b]
  \centering
  \includegraphics{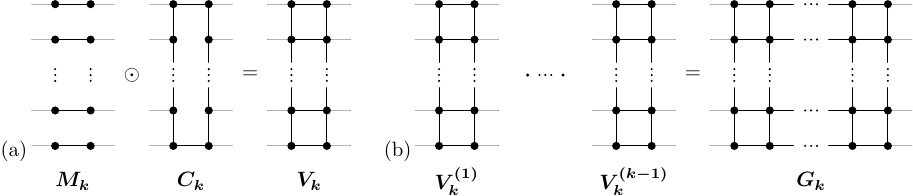}
  \caption{(a) Construction of the graphs $\boldsymbol{V}_k$. (b) Construction of the graph $\boldsymbol{G}_k$.}
  \label{fig:grid}
\end{figure}

We now need a standard result from graph minor theory that follows from \cite{robertson-seymour-3}:

\begin{thm}\label{thm:graph-minor-theory}
    For every planar graph $G$, there is a natural number $n_G$ such that $G$ is a minor of the $n_G \times n_G$ grid. 
\end{thm}

The following corollary is now immediate from \cref{lem:planar-class}, \cref{lem:minor-closed}, \cref{lem:unbounded-tw} and \cref{thm:graph-minor-theory}:

\begin{cor}\label{cor:alt-proof}
     The set $\calP = \bigcup_{k \in \mathbb{N}}\soe(\calP_k)$ is the set of all planar graphs.  
\end{cor}

Recalling the discussion preceding \cref{lem:planar-class}, this completes the proof of \cref{thm:quant-iso}, and thus gives us the promised alternative proof of the result of~\cite{david-laura}. 

\section{Exact Feasibility of NPA in Randomized Polynomial Time}
\label{sec:alg-aspects}
Being a semidefinite program, the NPA relaxation can be solved using standard techniques such as the ellipsoid method. 
However, such techniques can, in polynomial time, only decide the approximate feasibility of a system.
In this section, we use homomorphism indistinguishability to give a randomized algorithm for deciding whether any level of the NPA hierarchy is feasible exactly.
\thmAlgorithm*

By a recent result \cite[Theorem~1.1]{seppelt_algorithmic_2024},
homomorphism indistinguishability over every minor-closed graph class of bounded treewidth can be decided in randomized polynomial time.
In \cref{thm:main-theorem}, we have established that the $k^{\text{th}}$-level of the NPA hierarchy for the $(G,H)$-isomorphism game is feasible for two graphs $G$ and $H$ if and only if 
they are homomorphism indistinguishable over $\mathcal{P}_k$.
By \cref{lem:minor-closed,lem:bounded-tw},
the graph class $\mathcal{P}_k$ is minor-closed and of bounded treewidth. 
Hence, by \cite[Theorem~1.1]{seppelt_algorithmic_2024}, 
the feasibility of the $k^{\text{th}}$-level of the NPA hierarchy can be decided in randomized polynomial time for every fixed~$k$.
However, this result assumes $k$ to be fixed and not part of the input.
\Cref{thm:algorithmic-aspects} makes the dependence on $k$ effective.

\subsection{Modular Homomorphism Indistinguishability}
\label{sec:mod-hom}
As a first step towards \cref{thm:algorithmic-aspects}, we show that it can be decided in deterministic polynomial time whether $G$ and $H$ are \emph{homomorphism indistinguishable over $\mathcal{P}_k$ modulo a prime~$p$}, i.e.\ $\hom(F, G) \equiv \hom(F, H) \mod p$ for all $F \in \mathcal{P}_k$. 
We choose to work over a finite field in order to avoid memory issues with too large integers.

\begin{thm}\label{thm:algorithmic-aspects-prime}
    There exists a deterministic algorithm which decides, given graphs $G$ and $H$, an integer $k\geq 1$, and a prime $p$,
    whether $G$ and $H$ are homomorphism indistinguishable over $\mathcal{P}_k$ modulo~$p$.
    The algorithm runs in time $n^{O(k)} (k \log p)^{O(1)}$ for $n \coloneqq \max \{\lvert V(G) \rvert , \lvert V(H) \rvert  \}$.
\end{thm}

The algorithm in \cref{thm:algorithmic-aspects-prime} decides whether $G$ and $H$ are homomorphism indistinguishable over $\mathcal{P}_k$ by computing a basis $B$ for the $\mathbb{F}_p$-vector space $S$ spanned by homomorphism matrices of bilabelled graphs in $\mathcal{P}_k$. More precisely,
\[
    S \coloneqq \spn\left\{ \boldsymbol{P}_G\oplus \boldsymbol{P}_H \mid \boldsymbol{P} \in \mathcal{P}_k \right\} \subseteq \mathbb{F}_p^{(V(G)^k \cup V(H)^k) \times (V(G)^k \times V(H)^k)}
\]
where $\boldsymbol{P}_G \oplus \boldsymbol{P}_H \coloneqq \left( \begin{smallmatrix} \boldsymbol{P}_G & 0 \\ 0 & \boldsymbol{P}_H \end{smallmatrix} \right)$.
A basis $B \subseteq S$ can be computed iteratively as follows.
We initialise $B$ with the singleton set containing $\boldsymbol{J}_G \oplus \boldsymbol{J}_H$.
Subsequently, we repeatedly apply the operation from \cref{def:pk} to compute new vectors. 
Whenever a new vector is linearly independent from the vectors already in $B$, we add it to $B$.
Since the dimension of $S$ is at most $2n^{2k}$, this process terminates after a polynomial number of steps.
This procedure is formally described in \cref{alg-modular}.

In order to achieve a better runtime in \cref{thm:algorithmic-aspects-prime}, we give size-$O(k)$ sets $\mathcal{B}_k^P$ and $\mathcal{B}_k^S$ of bilabelled graphs generating $\mathcal{Q}_k^P$ and $\mathcal{Q}_k^S$.
Let $\mathcal{B}_k \coloneqq \mathcal{B}_k^P \cup \mathcal{B}_k^S$.
\begin{lem} \label{lem:basal-graphs}
	Let $k \geq 1$ and consider the following $(k,k)$-bilabelled graphs.
	\begin{itemize}
		\item $\boldsymbol{J} \coloneqq (J, (1, \dots, k), (k+1, \dots, 2k))$ with $V(J) = [2k]$ and $E(J) = \emptyset$,
		\item for $i \in [2k]$ and $j \coloneqq i+1$ if $i < 2k$ and $j \coloneqq 1$ otherwise, the graphs $\boldsymbol{C}^{=i}$ and $\boldsymbol{C}^{\sim i}$ which are obtained from $\boldsymbol{J}$ by, respectively, identifying or connecting the vertices $i$ and $j$,
		\item $\boldsymbol{I} \coloneqq (I, (1, \dots, k), (1, \dots, k))$ with $V(I) = [k]$ and $E(I) = \emptyset$,
		\item for $i \in [k]$, the graphs $\boldsymbol{M}^{\neq i} = (M^{ \neq i}, (1, \dots, k), (1, \dots, i-1, i', i+1, \dots,k))$ with $V(M^{\neq i}) = [k] \cup \{i'\}$, $E(M^{\neq i}) = \emptyset$ and $\boldsymbol{M}^{\sim i} = (M^{\sim i}, (1, \dots, k), (1, \dots, i-1, i', i+1, \dots, k))$ with $V(M^{\sim i}) = [k] \cup \{i'\}$, $E(M^{\sim i}) = \{ii'\}$.
	\end{itemize}
	Then $\mathcal{Q}_k^P$ is generated by $\mathcal{B}_k^P \coloneqq \{\boldsymbol{J}\} \cup \{\boldsymbol{C}^{=i}, \boldsymbol{C}^{\sim i} \mid i \in [2k]\}$ under parallel composition
	and $\mathcal{Q}_k^S$ is contained by the graph class generated by $\mathcal{B}_k^S \coloneqq \{\boldsymbol{I}\} \cup \{\boldsymbol{M}^{\neq i}, \boldsymbol{M}^{\sim i} \mid i \in [k]\}$ under series composition.
\end{lem}
\begin{proof}
	For the claim regarding $\mathcal{Q}_k^P$, 
	let $\boldsymbol{C}$ be as in \cref{def:atomic-q}.
	Let $\boldsymbol{F} \leq \boldsymbol{C}$ be a bilabelled minor of $\boldsymbol{C}$.
	Then there exists a partition $[2k] = A \sqcup I \sqcup J$ such that 
	\begin{itemize}
		\item for $i \in A$ and $j \coloneqq i+1$ if $i < 2k$ and $j \coloneqq 1$, the edge $ij$ in $\boldsymbol{C}$ is retained to obtain $\boldsymbol{F}$,
		\item for $i \in I$ and $j \coloneqq i+1$ if $i < 2k$ and $j \coloneqq 1$, the edge $ij$ in $\boldsymbol{C}$ is contracted to obtain $\boldsymbol{F}$,
		\item for $i \in J$ and $j \coloneqq i+1$ if $i < 2k$ and $j \coloneqq 1$, the edge $ij$ in $\boldsymbol{C}$ is deleted to obtain $\boldsymbol{F}$.
	\end{itemize}
	Hence, $\boldsymbol{F} = \boldsymbol{J} \odot \bigodot_{i\in I} \boldsymbol{C}^{=i} \odot \bigodot_{i \in A} \boldsymbol{C}^{\sim i}$.
	
	For the claim regarding $\mathcal{Q}_k^S$,
	let $\boldsymbol{M}$ be as in \cref{def:atomic-q}.
	Let $\boldsymbol{F} \leq \boldsymbol{M}$ be a bilabelled minor of $\boldsymbol{M}$.
	Then there exists a partition $[k] = A \sqcup I \sqcup J$ such that 
	\begin{itemize}
		\item for $i \in A$ and $j \coloneqq i+k$, the edge $ij$ is in $\boldsymbol{M}$ is retained to obtain $\boldsymbol{F}$,
		\item for $i \in I$ and $j \coloneqq i+k$, the edge $ij$ is in $\boldsymbol{M}$ is contracted to obtain $\boldsymbol{F}$,
		\item for $i \in J$ and $j \coloneqq i+k$, the edge $ij$ is in $\boldsymbol{M}$ is deleted to obtain $\boldsymbol{F}$.
	\end{itemize}
	Hence, $\boldsymbol{F} = \boldsymbol{I} \cdot \prod_{i \in J} \boldsymbol{M}^{\neq i} \cdot \prod_{i \in A} \boldsymbol{M}^{\sim i}$.
\end{proof}

\begin{algorithm}[t]
	\SetAlgoNoEnd
	\LinesNumbered
	\KwIn{Graphs $G$ and $H$, an integer $k \geq 1$, a prime $p$ in binary.}
	\KwOut{Whether $G$ and $H$ are homomorphism indistinguishable over $\mathcal{P}_k$ modulo~$p$.}
	
	for every $\boldsymbol{A} \in \mathcal{B}_k$, 
	compute the homomorphism matrices $\boldsymbol{A}_G \in \mathbb{F}_p^{V(G)^k \times V(G)^k}$ and $\boldsymbol{A}_H \in \mathbb{F}_p^{V(H)^k \times V(H)^k}$\;
	
	initialise $B \leftarrow \{ \boldsymbol{J}_G \oplus \boldsymbol{J}_H  \} \subseteq \mathbb{F}_p^{(V(G)^k \cup V(H)^k) \times (V(G)^k \cup V(H)^k)}$\;	
	
	\ForEach{$\boldsymbol{A} \in \mathcal{B}_k$}{
		\If{$\boldsymbol{A}_G \oplus \boldsymbol{A}_H \not\in \spn(B)$\nllabel{line:lindin1}}{ add $\boldsymbol{A}_G \oplus \boldsymbol{A}_H$ to $B$\;}
	}
	
	\Repeat{$B$ is not updated}{\nllabel{line:repeat}

		\ForEach{$\boldsymbol{A} \in \mathcal{B}_k^P$, $v \in B$}{
			$w \leftarrow (\boldsymbol{A}_G \oplus \boldsymbol{A}_H) \odot v$\;
			
			\If{$w \not\in \spn(B)$\nllabel{line:lindin2}}{ add $w$ to $B$\;}
		}

		\ForEach{$v_1, v_2 \in B$}{
			$w \leftarrow v_1 \cdot v_2$\;	
			\If{$w \not\in \spn(B)$\nllabel{line:lindin3}}{add $w$ to $B$\;}
		}
		\ForEach{$\sigma \in \cyclicpermutations$, $v \in B$}{
			$w \leftarrow v^{\sigma}$\;	
			\If{$w \not\in \spn(B)$\nllabel{line:lindin4}}{add $w$ to $B$\;}
		}	
	}
	
	\eIf{$\boldsymbol{1}_G^T v \boldsymbol{1}_G = \boldsymbol{1}_H^Tv \boldsymbol{1}_H$ for all $v \in B$ }{accept\;}{reject\;}
	
	\caption{Modular NPA.}
	\label{alg-modular}
\end{algorithm}

\begin{lem}
    \cref{alg-modular} is correct.
\end{lem}
\begin{proof}
    Let $B$ denote the set computed when \cref{alg-modular} terminates.
    We first argue that the vectors in $B$ span $S$.
    Clearly, $B \subseteq S$ and thus $\spn(B)\subseteq S$.
    The inclusion $S \subseteq \spn(B)$ 
    is shown by induction on the operations used in \cref{def:pk} to define $\mathcal{P}_k$.
    The simplest graph in that set is~$\boldsymbol{J}$.
    It holds that $\boldsymbol{J}_G \oplus \boldsymbol{J}_H \in B$, by initialisation.
    By \cref{lem:basal-graphs}, it holds that the homomorphism matrices of all bilabelled graphs in $\mathcal{Q}_k$ are in $\spn(B)$.
    For the inductive step, let $\boldsymbol{P} \in \mathcal{P}_k$ be some graph. Distinguish cases:

    If $\boldsymbol{P} = \boldsymbol{R} \cdot \boldsymbol{S}$ for the graphs $\boldsymbol{R}, \boldsymbol{S} \in \mathcal{P}_k$ 
    to which the inductive hypothesis applies then $\boldsymbol{R}_G \oplus \boldsymbol{R}_H, \boldsymbol{S}_G \oplus \boldsymbol{S}_H \in \spn(B)$.
    Hence, we may write $\boldsymbol{R}_G \oplus \boldsymbol{R}_H = \sum \alpha_i b_i$ and  $\boldsymbol{R}_G \oplus \boldsymbol{R}_H = \sum \beta_i b_i$ for some $\alpha_i, \beta_i \in \mathbb{F}_p$ and $b_i \in B$.
    Since the algorithm terminated, it holds that $b_i \cdot b_j \in \spn(B)$ for all $i,j$. 
    Thus, $\boldsymbol{P}_G \oplus \boldsymbol{P}_H = \sum \alpha_i \beta_j b_i \cdot b_j \in \spn(B)$.

    The cases when $\boldsymbol{P} = \boldsymbol{R} \odot \boldsymbol{S}$ for $\boldsymbol{R} \in \mathcal{Q}_k^P$ and when $\boldsymbol{P} = \boldsymbol{R}^\sigma$ for $\sigma \in \cyclicpermutations$ follow analogously.

    It remains to argue that the acceptance condition $\boldsymbol{1}_G^T v \boldsymbol{1}_G = \boldsymbol{1}_H^Tv \boldsymbol{1}_H$ for all $v \in B$  is correct.
    To that end, first assume that $G$ and $H$ are homomorphism indistinguishable over $\mathcal{P}_k$.
    Then for every $\boldsymbol{P} \in \mathcal{P}_k$ it holds that
    \[
         \boldsymbol{1}_G^T (\boldsymbol{P}_G \oplus \boldsymbol{P}_H) \boldsymbol{1}_G   
        = \soe(\boldsymbol{P}_G) 
        = \soe(\boldsymbol{P}_H)
        = \boldsymbol{1}_H^T (\boldsymbol{P}_G \oplus \boldsymbol{P}_H) \boldsymbol{1}_H.  
    \]
    Since $S$ is spanned by these vectors, it follows that $\boldsymbol{1}_G^T v \boldsymbol{1}_G = \boldsymbol{1}_H^Tv \boldsymbol{1}_H$ for all $v \in B$.

    Conversely, suppose that $\boldsymbol{1}_G^T v \boldsymbol{1}_G = \boldsymbol{1}_H^Tv \boldsymbol{1}_H$  holds for all $v \in B$.
    Let $\boldsymbol{P} \in \mathcal{P}_k$ be arbitrary.
    By the previous claim, $\boldsymbol{P}_G \oplus \boldsymbol{P}_H \in \spn(B)$.
    Hence, the assumption implies that $\soe(\boldsymbol{P}_G) = \soe(\boldsymbol{P}_H)$, as desired.
\end{proof}

\begin{lem} \label{lem:runtime}
    \cref{alg-modular} terminates in time $n^{O(k)} (k \log p)^{O(1)}$ for $n \coloneqq \max \{\lvert V(G) \rvert , \lvert V(H) \rvert  \}$.
\end{lem}
\begin{proof}
	The initialisation (above \cref{line:repeat}) can be completed in time $n^{O(k)} k^{O(1)} (\log p)^{O(1)}$ noting that the set $\mathcal{B}_k$ is of size $O(k)$ by \cref{lem:basal-graphs}.
	
    Throughout the execution of the algorithm, $B$ is a set of linearly independent vectors in a $2n^{2k}$-dimensional vector space over $\mathbb{F}_p$.
    This implies that the loop in \cref{line:repeat} is entered at most $2n^{2k}$ many times.
    It further more implies that the checks for linear independence in \cref{line:lindin1,line:lindin2,line:lindin3,line:lindin4} can be carried out in $n^{O(k)} (\log p)^{O(1)}$.

	In each iteration of in \cref{line:repeat}, 
	the first inner loop is entered at most $n^{O(k)} k^{O(1)}$ times.
	The second inner loop is entered at most $n^{O(k)}$ times.
	The third inner loops is entered at most $n^{O(k)} k^{O(1)}$ times.
	This yields the overall runtime bound.
\end{proof}

\subsection{Reducing NPA to Modular NPA}

If $G \not\cong_{\mathcal{P}_k} H$, then there exists a graph $F \in \mathcal{P}_k$ such that $\hom(F, G) \neq \hom(F, H)$.
Since $\hom(F, G), \hom(F, H) \leq n^{\mid V(F)\mid }$ for $n \coloneqq \max\{\lvert V(G) \rvert , \lvert V(H) \rvert \}$, 
it follows that $G$ and $H$ are also not homomorphism indistinguishable over $\mathcal{P}_k$ modulo every prime $p$ greater than $n^{\mid V(F)\mid }$.
Unfortunately, there is a priori no bound on the size of $F$ in terms of $n$.
In this section, we give such a bound and thereby derive \cref{thm:algorithmic-aspects} from \cref{thm:algorithmic-aspects-prime}.
For $l \in \mathbb{N}$, 
write $(\mathcal{P}_k)_{\leq l} \coloneqq \{ F \in \mathcal{P}_k \mid \lvert V(F)\lvert  \leq l\}$.

\begin{thm} \label{thm:witness-function}
	Let $k \geq 1$.
	Let $G$ and $H$ be a graphs on at most $n$ vertices.
	Let $f_k(n) \coloneqq 2k \cdot 4^{n^{2k}}$.
	Then 
	\[
		G \cong_{\mathcal{P}_k} H \iff
		G \cong_{(\mathcal{P}_k)_{\leq f_k(n)}} H.
	\]
\end{thm}

Towards \cref{thm:witness-function},
we define  the following complexity measure $\nu \colon \mathcal{P}_k \to \mathbb{N}$ inductively. If $\boldsymbol{Q} \in \mathcal{Q}_k$, then $\nu(\boldsymbol{Q}) \coloneqq 1$.
For $\boldsymbol{F} \in \mathcal{P}_k$, define $\nu(\boldsymbol{F})$ inductively as the least number $n \in \mathbb{N}$ such that there exist
\begin{enumerate}
	\item $\boldsymbol{F}' \in \mathcal{P}_k$ and $\boldsymbol{Q} \in \mathcal{Q}_k^P$ such that $\boldsymbol{F} = \boldsymbol{Q} \odot \boldsymbol{F}'$ and $n = \nu(\boldsymbol{F}')$, or
	\item $\boldsymbol{F}', \boldsymbol{F}'' \in \mathcal{P}_k$ such that $\boldsymbol{F} = \boldsymbol{F}' \cdot \boldsymbol{F}''$ and $n = \max\{ \nu(\boldsymbol{F}'), \nu(\boldsymbol{F}'')\} +1$, or
	\item $\boldsymbol{F}' \in \mathcal{P}_k$ and $\sigma \in \cyclicpermutations$ such that $\boldsymbol{F} = (\boldsymbol{F}')^\sigma$ and $n = \nu(\boldsymbol{F}')$.
\end{enumerate}
By \cref{def:pk}, $\nu \colon \mathcal{P}_k \to \mathbb{N}$ is well-defined.
\begin{lem} \label{lem:nu-size}
	Let $k \geq 1$.
	For every $\boldsymbol{F} \in \mathcal{P}_k$, it holds that $\boldsymbol{F}$ has at most $2k \cdot 2^{\nu(\boldsymbol{F})}$ vertices.
\end{lem}
\begin{proof}
	The proof is by structural induction on $\boldsymbol{F} \in \mathcal{P}_k$.
	If $\boldsymbol{F} \in \mathcal{Q}_k$ then $\boldsymbol{F}$ has at most $2k$ vertices and $\nu(\boldsymbol{F}) = 1$.
	For the inductive step, let $\boldsymbol{F} \in \mathcal{P}_k$ and distinguishing cases:
	\begin{enumerate}
		\item There exist $\boldsymbol{F}' \in \mathcal{P}_k$ and $\boldsymbol{Q} \in \mathcal{Q}_k^P$ such that $\boldsymbol{F} = \boldsymbol{Q} \odot \boldsymbol{F}'$ and $\nu(\boldsymbol{F}) = \nu(\boldsymbol{F}')$.
		
		Since taking the parallel composition of $\boldsymbol{F}'$ with the fully labelled graph $\boldsymbol{Q}$ does  not increase the number of vertices in $\boldsymbol{F}'$, it follows that $\boldsymbol{F}$ has at most $2k \cdot 2^{\nu(\boldsymbol{F}')} = 2k \cdot 2^{\nu(\boldsymbol{F})}$ many vertices.
		
		\item There exist $\boldsymbol{F}', \boldsymbol{F}'' \in \mathcal{P}_k$ such that $\boldsymbol{F} = \boldsymbol{F}' \cdot \boldsymbol{F}''$ and $\nu(\boldsymbol{F}) = \max\{ \nu(\boldsymbol{F}'), \nu(\boldsymbol{F}'')\} +1$.
		
		The number of vertices in $\boldsymbol{F}$ is at most the number of vertices in $\boldsymbol{F}'$ plus the number of vertices in $\boldsymbol{F}''$.
		Hence, $\boldsymbol{F}$ has at most $2k \cdot (2^{\nu(\boldsymbol{F}')} + 2^{\nu(\boldsymbol{F}'')}) \leq 2k \cdot 2 \cdot 2^{\nu(\boldsymbol{F}) - 1} = 2k \cdot 2^{\nu(\boldsymbol{F})}$ many vertices.
		
		\item There exist $\boldsymbol{F}' \in \mathcal{P}_k$ and $\sigma \in \cyclicpermutations$ such that $\boldsymbol{F} = (\boldsymbol{F}')^\sigma$ and $\nu(\boldsymbol{F}) = \nu(\boldsymbol{F}')$.
		
		The number of vertices in $\boldsymbol{F}'$ is not affected by permuting labels. Hence, $\boldsymbol{F}$ has at most $2k \cdot 2^{\nu(\boldsymbol{F})}$ many vertices. \qedhere
	\end{enumerate} 
\end{proof}

As in \cref{sec:mod-hom}, we consider a sequence of nested spaces of homomorphism matrices. For $l \geq 1$, write
\[
	S_l \coloneqq \spn \{ \boldsymbol{F}_G \oplus \boldsymbol{F}_H \mid \boldsymbol{F} \in \mathcal{P}_k, \nu(\boldsymbol{F}) \leq l \} 
	\subseteq \mathbb{R}^{(V(G)^k \cup V(H)^k) \times (V(G)^k \cup V(H)^k)}.
\]
Clearly, $S_1 \subseteq S_2 \subseteq \dots \subseteq S \coloneqq \bigcup_{l \geq 1} S_l$.
The space $S$ is of dimension at most $2n^{2k}$.
The following lemma shows that $S_{2n^{2k}} = S$.

\begin{lem} \label{lem:chain-collapse}
	If $l \geq 1$ is such that $S_l = S_{l+1}$,
	then $S_l = S$.
\end{lem}
\begin{proof}
	Since $S_l \subseteq S$, it suffices to show the converse inclusion.
	We show by structural induction on $\boldsymbol{F} \in \mathcal{P}_k$ that $\boldsymbol{F}_G \oplus \boldsymbol{F}_H \in S_l$.
	If $\boldsymbol{F}$ is atomic then the claim follows immediately.
	For the inductive step, distinguish cases:
	\begin{enumerate}
		\item There exist $\boldsymbol{F}' \in \mathcal{P}_k$ and $\boldsymbol{Q} \in \mathcal{Q}_k^P$ such that $\boldsymbol{F} = \boldsymbol{Q} \odot \boldsymbol{F}'$.
		
		By the inductive hypothesis, $\boldsymbol{F}'_G \oplus \boldsymbol{F}'_H \in S_l$.
		Hence, there exist $\boldsymbol{F}^i \in \mathcal{P}_k$ with $\nu(\boldsymbol{F}^i) \leq l$ and coefficients $\alpha_i \in \mathbb{R}$ such that $\boldsymbol{F}'_G \oplus \boldsymbol{F}'_H  = \sum \alpha_i \boldsymbol{F}^i_G \oplus \boldsymbol{F}^i_H$.
		By definition, $\nu(\boldsymbol{Q} \odot \boldsymbol{F}^i) \leq \nu(\boldsymbol{F}^i) \leq l$.
		Hence,
		\[
			\boldsymbol{F}_G \oplus \boldsymbol{F}_H = (\boldsymbol{Q}_G \oplus \boldsymbol{Q}_H) \odot \sum \alpha_i \boldsymbol{F}^i_G \oplus \boldsymbol{F}^i_H
			= \sum \alpha_i (\boldsymbol{Q} \odot \boldsymbol{F}^i)_G
			 \oplus (\boldsymbol{Q} \odot \boldsymbol{F}^i)_H
			 \in S_l.
		\]
		
		\item There exist $\boldsymbol{F}', \boldsymbol{F}'' \in \mathcal{P}_k$ such that $\boldsymbol{F} = \boldsymbol{F}' \cdot \boldsymbol{F}''$.
		
		By the inductive hypothesis, $\boldsymbol{F}'_G \oplus \boldsymbol{F}'_H,\boldsymbol{F}''_G \oplus \boldsymbol{F}''_H  \in S_l$.
		Hence, there exist $\boldsymbol{F}^i \in \mathcal{P}_k$, $\boldsymbol{K}^j \in \mathcal{P}_k$ with $\nu(\boldsymbol{F}^i),\nu(\boldsymbol{K}^j) \leq l$ and coefficients $\alpha_i, \beta_j \in \mathbb{R}$ such that $\boldsymbol{F}'_G \oplus \boldsymbol{F}'_H  = \sum \alpha_i \boldsymbol{F}^i_G \oplus \boldsymbol{F}^i_H$
		and $\boldsymbol{F}''_G \oplus \boldsymbol{F}''_H  = \sum \beta_j \boldsymbol{K}^j_G \oplus \boldsymbol{K}^j_H$.
		By definition, $\nu(\boldsymbol{F}^i \cdot \boldsymbol{K}^j) \leq \max\{ \nu(\boldsymbol{F}^i), \nu(\boldsymbol{K}^j)\} \leq l+1$ for all $i, j$.
		Hence, 
		\[
			\boldsymbol{F}_G \oplus \boldsymbol{F}_H = \sum \alpha_i \beta_j (\boldsymbol{F}^i \cdot \boldsymbol{K}^j)_G \oplus 
			(\boldsymbol{F}^i \cdot \boldsymbol{K}^j)_H
			\in S_{l+1}  = S_l.
		\]
		
		\item There exist $\boldsymbol{F}' \in \mathcal{P}_k$ and $\sigma \in \cyclicpermutations$ such that $\boldsymbol{F} = (\boldsymbol{F}')^\sigma$.
		
		By the inductive hypothesis, $\boldsymbol{F}'_G \oplus \boldsymbol{F}'_H \in S_l$.
		Hence, there exist $\boldsymbol{F}^i \in \mathcal{P}_k$ with $\nu(\boldsymbol{F}^i) \leq l$ and coefficients $\alpha_i \in \mathbb{R}$ such that $\boldsymbol{F}'_G \oplus \boldsymbol{F}'_H  = \sum \alpha_i \boldsymbol{F}^i_G \oplus \boldsymbol{F}^i_H$.
		By definition, $\nu((\boldsymbol{F}^i)^\sigma) \leq \nu(\boldsymbol{F}^i) \leq l$.
		Hence,
		\[
		\boldsymbol{F}_G \oplus \boldsymbol{F}_H = \left( \sum \alpha_i \boldsymbol{F}^i_G \oplus \boldsymbol{F}^i_H \right)^\sigma
		= \sum \alpha_i ((\boldsymbol{F}^i)^\sigma)_G \oplus ((\boldsymbol{F}^i)^\sigma)_H
		\in S_l. \qedhere
		\] 
	\end{enumerate} 
\end{proof}

This concludes the preparations for the proof of \cref{thm:witness-function}.
\begin{proof}[Proof of \cref{thm:witness-function}]
	Only the backward implication requires a justification.
	Suppose that $G$ and $H$ are homomorphism indistinguishable over all graphs in $\mathcal{P}_k$ of size at most $2k \cdot 4^{n^{2k}}$.
	
	Let $\boldsymbol{F} \in \mathcal{P}_k$ be of arbitrary size.
	By \cref{lem:chain-collapse}, there exist $\boldsymbol{F}^i \in \mathcal{P}_k$ with $\nu(\boldsymbol{F}^i) \leq 2n^{2k}$ and coefficients $\alpha_i \in \mathbb{R}$ such that
	$\boldsymbol{F}_G \oplus \boldsymbol{F}_H = \sum \alpha_i \boldsymbol{F}^i_G \oplus \boldsymbol{F}^i_H$.
	By \cref{lem:nu-size}, the $\boldsymbol{F}^i$ have at most $2k \cdot 4^{n^{2k}}$ many vertices.
	Thus, $ \hom(\soe(\boldsymbol{F}^i), G) =  \hom(\soe(\boldsymbol{F}^i), H)$, by assumption.
	It follows that
	\begin{align*}
		\hom(\soe(\boldsymbol{F}), G)
		&= \boldsymbol{1}_G^T (\boldsymbol{F}_G \oplus \boldsymbol{F}_H) \boldsymbol{1}_G \\
		&= \sum \alpha_i \boldsymbol{1}_G^T (\boldsymbol{F}^i_G \oplus \boldsymbol{F}^i_H) \boldsymbol{1}_G \\
		&= \sum \alpha_i \hom(\soe(\boldsymbol{F}^i), G) \\
		&= \sum \alpha_i \hom(\soe(\boldsymbol{F}^i), H) \\
		&= \sum \alpha_i \boldsymbol{1}_H^T (\boldsymbol{F}^i_G \oplus \boldsymbol{F}^i_H) \boldsymbol{1}_H \\
		&= \hom(\soe(\boldsymbol{F}), H).
	\end{align*}	
	Here, $\boldsymbol{1}_G, \boldsymbol{1}_H \in \mathbb{R}^{V(G)^k \cup V(H)^k}$ denote the indicator vectors on the coordinates which correspond to $G$ and $H$, respectively.
	Hence, $G$ and $H$ are homomorphism indistinguishable over $\mathcal{P}_k$.
\end{proof}

It remains to derive \cref{thm:algorithmic-aspects} from \cref{thm:algorithmic-aspects-prime,thm:witness-function}. 

\begin{proof}[Proof of \cref{thm:algorithmic-aspects}]
	A randomized algorithm for the problem \cref{thm:algorithmic-aspects} proceeds as similar to \cite[Algorithm~2]{seppelt_algorithmic_2024}.
	Let $N \coloneqq 2k \cdot 4^{n^{2k}}$ be as in \cref{thm:witness-function}.
	For $\lceil 4 \log(N \log(n)) \rceil = n^{O(k)}$ times, sample an integer $N \log n < p \leq (N \log n)^2$.
	This integer requires $n^{O(k)}$ many bits.
	In time $n^{O(k)}$, deterministically check whether $p$ is a prime \cite{agrawal_primes_2004}.
	If it is a prime, run \cref{alg-modular} for $G$, $H$, $k$, and $p$.
	By \cref{lem:runtime}, this takes time $n^{O(k)}k^{O(1)}$.
	If the algorithm asserts that $G \not\cong_{\mathcal{P}_k} H$, reject.
	Otherwise, proceed.
	If $p$ is not a prime, proceed to the next iteration.

	By the proof of \cite[Theorem~1.1]{seppelt_algorithmic_2024},
	the probability of incorrectly accepting an instance such that $G \not\cong_{\mathcal{P}_k} H$ is less than $1/2$.
\end{proof}

\section{Discussion}\label{sec:conclusion}

We have established a characterization of the feasibility of the $k^{\text{th}}$-level of the NPA hierarchy of relaxations for the $(G, H)$-isomorphism game in terms of homomorphism indsitingushability over the graph class $\mathcal{P}_k$.
We know that $\mathcal{P}_k$ is a subclass of planar graphs, has bounded-treewidth, and is closed under taking minors.
We only have an inductive description of the class $\mathcal{P}_k$ as being generated from $\calQ_k$ using series composition, cyclic permutations, and parallel composition with atomic graphs from $\calQ_k^P$. Thus, a natural question is to understand better the properties of the graphs in $\mathcal{P}_k$.
Recall that $\mathcal{L}_k$ is the analogous class of graphs for the $k^{\text{th}}$-level of the Lasserre hierarchy of relaxations of the integer program for graph isomorphism.

\begin{prob}
  Is $\mathcal{P}_k$ equal to the intersection of $\mathcal{L}_k$ with the class of all planar bilabelled graphs?
\end{prob}

In the field of graph isomorphism testing, the $k$-dimensional Weisfeiler-Leman ($k$-WL) algorithm is an essential subroutine.
It iteratively constructs an automorphism-invariant partition of the $k$-tuples of vertices of a graph.
For some classes of graphs, the $k$-WL algorithm is sufficient to test isomorphism.
This is closely related to the feasibility of the $k^{\text{th}}$-level of the Lassere hierarchy of SDP relaxations of the graph isomorphism integer program~\cite{roberson-seppelt-arxiv}.
This motivates the following problem.

\begin{prob}
  Is there a class of graphs $\mathcal{C}$ for which the $k^{\text{th}}$-level of the NPA hierarchy determines quantum isomorphism and isomorphism does not coincide with quantum isomorphism on~$\mathcal{C}$?
\end{prob}

Note that, for example, it is known that two trees are isomorphic if and only if they are quantum isomorphic, and that $1$-WL determines if two trees are isomorphic. However, it would be interesting to find a class of graphs that can be quantum isomorphic without being isomorphic, and where we can decide quantum isomorphism by some fixed level of the NPA hierarchy.  

\bibliographystyle{quantum}
\bibliography{npa}

\appendix

\newpage

\section{Linear algebra}\label{app:linalg}

Recall that a linear map $\Phi: M_n(\mathbb{C}) \to M_k(\mathbb{C})$ is said to be \emph{positive} if it maps all positive semidefinite matrices in $M_n(\mathbb{C})$ to positive semidefinite matrices in $M_k(\mathbb{C})$. $\Phi$ is said to be \emph{completely positive} if for all $l \in \mathbb{N}$ the map $\mathbb{I}_l \otimes \Phi : M_n(\complex) \otimes M_l(\complex) \to M_k(\complex) \otimes M_l(\complex)$ is positive, where $\mathbb{I}_l: M_l(\mathbb{C}) \to M_l(\complex)$ is the identity map. 

The \emph{Choi matrix} of a map $\Phi: M_n(\mathbb{C}) \to M_k(\mathbb{C})$ is the $nk \times nk$ matrix $\sum_{i,j = 1}^k E_{ij}\otimes \Phi(E_{ij})$, where $E_{ij}$ are the matrix units. It follows from a well known result of Choi \cite{CHOI1975285} that a linear map $\Phi: M_n(\mathbb{C}) \to M_k(\mathbb{C})$ is completely positive if and only if its Choi matrix is positive semidefinite. 

\begin{lem}[{\cite[Lemma A.1]{david_mathprog}}]
  \label{lem:linalg_unitary}
  Let $\calA$ and $\calB$ be self-adjoint unital subalgebras of $\complex^{n\times n}$ and $\varphi\colon \calA \to \calB$ be a trace-preserving isomorphism such that $\varphi(X^*) = \varphi(X)^*$ for all $X\in\calA$.
  Then there exists a unitary $U \in \complex^{n\times n}$ such that $\varphi(X) = UXU^*$ for all $X\in\calA$.
\end{lem}


\begin{lem}[Generalisation of {\cite[Lemma 4.9]{david_mathprog}}]
  \label{lem:lingalg_series}
  Let $\Phi_1, \Phi_2\colon M_n(\mathbb{C}) \to M_n(\mathbb{C})$ be two linear maps which are completely positive, trace-preserving, and unital.
  Then for any matrices $X$ and $Y$ such that $\Phi_1(X) = Y$ and $\Phi_2(Y) = X$ it holds that $\Phi_1(XW) = Y\Phi_1(W)$ and $\Phi_1(WX) = \Phi_1(W)Y$ for all $W \in M_n(\mathbb{C})$. Furthermore, 
  
  \begin{enumerate}[label = \roman*.]
      \item $X$ and $Y$ are cospectral in this case,
      \item and $\Phi_1^*(Y) = X$. 
  \end{enumerate}
\end{lem}
\begin{proof}
    The proof we present is a minor modification of the proof given in \cite[lemma 4.10]{david_mathprog}. Since $\Phi_1$ is completely positive, it admits a Kraus decomposition $\Phi_1(W) = \sum_{i=1}^m K_i W K_i^{*}$, for some matrices $K_i \in M_n(\mathbb{C})$. Note that since $\Phi_1$ is trace-preserving and unital, we have that $\sum_{i}K_iK_i^* = \sum_{i}K_i^*K_i = I$ and similarly for the $L_i$. The main idea of the proof is to show that $K_i X = YK_i$ for all $i$. Given this, it is easy to see that $\Phi_1(XW) = \sum_{i}K_i XW K_i^* = Y \sum_{i}K_iW K_i^* = Y \Phi_1(W)$. Set $\mathcal{Z} = \sum_{i} (K_i X - YK_i)(K_iX - YK_i)^*$. We note that 

    \begin{align*}
        \mathcal{Z} & = \sum_{i} (K_i X - YK_i)(K_iX - YK_i)^* \\
        & = \sum_i K_i XX^* K_i^* - \sum_i K_i X K_i^* Y^* - \sum_i Y K_i X K_i^* + \sum_i Y K_iK_i^* Y^* \\
        & = \Phi_1(XX^*) - \Phi_1(X)Y^* - Y \Phi_1(X^*) + YY^* \\
        & = \Phi_1(XX^*) - YY^*, 
    \end{align*}
    where we use the definition of $\Phi_1$ and that it is unital, and also the fact that $\Phi_1(W^*) = \Phi_1(W)^*$ for all $W \in M_n(\mathbb{C}$. 
    
    Since $\mathcal{Z}$ is positive semidefinite, we have $0 \leq \tr(\mathcal{Z}) = \tr(\Phi_1(XX^*) - YY^*) = \tr(XX^* - YY^*)$, where in the last equality we used the fact that $\Phi_1$ is trace preserving. In a similar fashion, but using $\Phi_2$, 
    we can obtain that $0 \le \tr(YY^* - XX^*)$. Hence, $\tr(\mathcal{Z}) = 0$ and therefore $\mathcal{Z} = 0$ since it is positive semidefinite. It follows that $K_iX = YK_i$ for all $i$ as desired.

    Thus we have shown that $\Phi_1(XW) = Y\Phi_1(W)$ for all $W \in M_n(\mathbb{C})$. To see that $\Phi_1(WX) = \Phi_1(W)Y$, we note that by assumption $\Phi_1(X^*) = Y^*$ and $\Phi_2(Y^*) = X^*$. Thus by the same argument as above, we have that $K_iX^* = Y^*K_i$, and by taking adjoints we obtain $XK_i^* = K_i^*Y$.
    

    For the final claim, let us denote the self-adjoint subalgebra of $M_n(\mathbb{C})$ generated by $X$ and $Y$ as $\calA_X$ and $\calA_Y$ respectively. Then, the restriction of $\Phi_1$ to $\calA_X$ is a trace-preserving isomorphism from $\calA_X$ to $\calA_Y$, which also preserves adjoints. Hence, it follows from \cref{lem:linalg_unitary} that $Y = \Phi_1(X) = UXU^*$ for some unitary $U \in M_n(\mathbb{C})$, so that $X$ and $Y$ are cospectral. Now, consider $\Phi_1^*(Y) = \sum_i K_i^* Y K_i =  \sum_i  X K_i^* K_i = X$, since $K_i^* Y = X K_i$. 
\end{proof}

\begin{lem}[{\cite[Lemma 4.5]{david_mathprog}}]
  \label{lem:linalg_parallel}
  Let $D\in\complex^{m\times n}$ be a matrix and $u\in\complex^n$ and $v \in \complex^{m}$.
  Then the following are equivalent:
  \begin{enumerate}
    \item $D(u\odot w) = v\odot(Dw)$ for all $w\in\complex^n$,
    \item $D_{ij} = 0$ whenever $v_i \neq u_j$,
    \item $D^*(v\odot z) = u\odot(D^*z)$ for all $z\in\complex^m$.
  \end{enumerate}
\end{lem}

\section{The NPA Hierarchy}\label{sec:npa}

Our treatment of the NPA hierarchy is based on~\cite[Chapter 8]{watrous_notes}.
Let $(p(a,b|x,y))_{x \in X, y \in Y, a \in A, b \in B}$ be a correlation in $P(X, Y, A, B)$. We now address the issue of determining if $p$ is a quantum correlation, i.e. we ask if there are commuting PVMs $\{E_{xa}\}_{a \in A}$ and $\{F_{y,b}\}_{b \in B}$ for each $x \in X$ and $y \in Y$, and a state $\psi$ such that $p(a,b|x,y) = \braket{ \psi, E_{xa} F_{yb} \psi}$ for all $x \in X$, $y \in Y$, $a \in A$ and $b \in B$. 

Note that in order to determine the correlation $p$, one only needs the vectors $\{E_{xa}\psi\}_{x\in X, y \in A}$ and $\{F_{yb}\psi\}_{y \in Y, b \in B}$. In particular, we note that all of this information is contained in the Gram Matrix of these vectors. 

Let us define $\Sigma = (X \times A) \sqcup (Y \times B)$. We can now define an SDP relaxation of the problem of determining whether $p$ is a quantum correlation, by searching for positive matrices $R$ indexed by $\Sigma^{\leq 1}$ such that $R_{xa, yb} = p(a,b|x,y)$ for all $x \in X, y \in Y, a \in A, b \in B$. A Gram Matrix arising from an actual quantum strategy as described earlier also satisfies several other equations. For example, one always has that $\sum_{a} E_{x,a}\psi = \psi$ for all $x \in X$. Hence, we may restrict our search space to positive matrices indexed by $\Sigma^{\leq 1}$ that satisfy these additional constraints. Each of these additional constraints actually turns out to be affine. Hence, we get an SDP relaxation of the problem of determining if $p$ is a quantum correlation as discussed earlier. 

One can construct increasingly complicated SDP relaxations of this problem in a similar manner by searching over positive matrices indexed by longer words formed from the alphabet $\Sigma$. For each $k \in \mathbb{N}$, one searches over positive semidefinite matrices indexed by words from $\Sigma^{\leq}$.

This hierarchy of SDP relaxations is known as the \emph{NPA hierarchy}, named after the authors of \cite{Navascues_2008} who introduced it, where the choice of $k$ corresponds to different levels of the hierarchy. We give more details in the next section. However, we shall restrict ourselves to the case of synchronous correlations and games. We refer the reader to \cite{Navascues_2008} for more details regarding the general case. 

\subsection{A Synchronous NPA Hierarchy}

We choose our alphabet to be $\Sigma = X \times A$ instead of $(X \times A) \sqcup (X\times A)$, since we are restricting ourselves to synchronous correlations. We define $\sim$ to be the finest equivalence relation satisfying the following two properties: 
\begin{enumerate}
    \item For each $x,a \in X \times A$, $s(x,a)(x,a)t \sim s(x,a)t$ for all $s,t \in \Sigma^*$.
    \item $st \sim ts$ for all words $s,t \in \Sigma^*$. 
\end{enumerate}

\begin{definition}
A function $\phi: \Sigma^* \to \mathbb{C}$ is said to be \emph{admissible} if:
\begin{enumerate}
    \item $\phi(\epsilon) = 1$
    \item For all words $s,t \in \Sigma^*$, we have 
    $$\sum_{a \in A}\phi(s(x,a)t) = \phi(st)$$
    for each $x \in X$.
    \item For all words $s,t \in \Sigma^*$, we have 
    $$\phi(s(x,a)(y,b)t) = 0$$
    for each $x,y \in X$ and $a,b \in A$ such that $V(a,b\mid x,y) = 0$.
    \item For all words $s,t \in \Sigma^*$ satisfying $s \sim t$, we have $\phi(s) = \phi(t)$.
\end{enumerate}

Similarly, an \emph{admissible function of order $k$} is a function $\phi: \Sigma^{\leq 2k} \to \mathbb{C}$ satisfying above conditions. An \emph{admissible operator of order $k$} is a positive semidefinite matrix $\mathcal{R} \in M_{\Sigma^{\leq k}}(\mathbb{C})$ such that there exists an admissible function $\phi: \Sigma^{\leq 2k}$ of order $k$ satisfying $\mathcal{R}_{s,t} = \phi(s^Rt).$ 
\end{definition}

We now state the result about the convergence of the NPA hierarchy for synchronous correlations from \cite{russell_synchronous_2023} without a proof: 

\begin{lem}[{\cite[Corollary 2]{russell_synchronous_2023}}]\label{lem:conv-npa-corr}
Let $p \in P(X,A)$ be a synchronous correlation. Then, $p \in C_q(X,A)$ if and only if for each $k \in \mathbb{N}$ there exists an admissible operator $\mathcal{R}^k$ of order $k$ such that $\mathcal{R}^k_{(x,a), (y,b)} = p(a,b\mid x,y)$ for all $x,y \in X, \ a,b \in A$.
\end{lem} \

\subsection{Proof of Convergence of NPA Hierarchy for Quantum Isomorphism}

We finish the proof of \ref{prop:conv-npa-graph} as promised in \cref{sec:preliminaries}. 

\begin{proof}[Proof of \ref{prop:conv-npa-graph}]
    The ``only if" direction follows from \cref{prop:perf-strat-graph}. We now show the ``if" direction. 

    Let the sequence of certificates be denoted $\Gamma^1, \Gamma^2, \dots$. We will establish the existence of certificates $\Gamma^{'1}, \Gamma^{'2}, \dots$ that also satisfy $\Gamma^{'i}_{(g,h), (g',h')}=  \Gamma^{'j}_{(g,h), (g',h')}$ for all $i, j \in \mathbb{N}$, $g,g' \in V(G)$ and $h,h' \in V(H)$. The proof is then complete by an application of \cref{lem:conv-npa-corr}. 

    We proceed in a similar fashion as in the proof of \cite[Corollary 2]{russell_synchronous_2023}. For each certificate $\Gamma^k$, we construct an infinite matrix $\widehat{\Gamma}^k$ indexed by the elements of $\Sigma^*$ as follows: 
    \begin{align*}
        \widehat{\Gamma}^k_{s, t} = \begin{cases}
            \Gamma^k_{s,t} & \text{ if } s,t \in \Sigma^{\leq k} \\
            0 & \text{ otherwise }
        \end{cases}
    \end{align*}

    We may now regard each $\widehat{\Gamma}_k$ as an element of $l^{\infty}(\Sigma^*)$. Since each element lies in the unit ball of $l^{\infty}(\Sigma^*)$ (as verified in the proof of \cite[Corollary 2]{russell_synchronous_2023}), it follows from the Banach-Alagolu theorem that the sequence $\{\widehat{\Gamma}_k\}_{k = 1}^{\infty}$ has a limit point (in the weak-$*$ topology) $\Gamma'$ that also lies in the unit ball of $l^{\infty}(\Sigma^*)$. 

    For each $k \in \mathbb{N}$, let $\Gamma^{'k}$ denote the square sub-matrix of $\Gamma'$ indexed by $\Sigma^{\leq k}$. It is clear that $\Gamma^{'i}_{(g,h), (g',h')}=  \Gamma^{'j}_{(g,h), (g',h')}$ for all $i, j \in \mathbb{N}$, $g,g' \in V(G)$ and $h,h' \in V(H)$. We only need to show that each $\Gamma^{'k}$ is a certificate for the $k^{\text{th}}$-level of the NPA hierarchy. 

    We start by showing that each $\Gamma^{'k}$ is positive semidefinite. Firstly, let us denote the subsequence of $\{\widehat{\Gamma}^k\}_{k=1}^{\infty}$ that converges to $\Gamma'$ by $\{\widehat{\Gamma}^{k_i}\}_{i=1}^{\infty}$. Let us consider a vector $v = \sum_{s\in \Sigma^{\leq}} v_s e_s$ expressed in the standard basis of $\mathbb{C}^{\Sigma^{\leq k}}$. We then have 
    \begin{align*}
        \braket{v , \Gamma^{'k} v} & = \sum_{s, t \in \Sigma^{\leq k}} \overline{v_{s}}v_t \braket{e_s, \Gamma^{'k} e_t} \\
        & = \sum_{s, t \in \Sigma^{\leq k}} \overline{v_{s}}v_t \Gamma^{'k}_{s,t} = \lim_{l \to \infty }\left(\sum_{s, t \in \Sigma^{\leq k}} \overline{v_{s}}v_t \widehat{\Gamma}^{k_l}_{s,t} \right) \\
        & = \lim_{l \to \infty } \braket{v , \widehat{\Gamma}^{k_l} v}  \geq 0 
    \end{align*}

    In the third equality, we use the weak-$*$ convergence of $\widehat{\Gamma}^k$ to $\Gamma'$ in the penultimate step, we identify $v$ with an appropriate vector in $\mathbb{C}^{\Sigma^{\leq k_l}}$. We also note that it is possible that $k \leq k_l$, for small $l$. In this case, we only evaluate the limit on the tail, i.e. on $k_l$ such that $k_l \geq k$. The last step now follows from the fact that each $\widehat{\Gamma}^{l_k}$ is positive semidefinite.

    We may also show that each $\Gamma^{'k}$ satisfies the algebraic conditions in \cref{def:cert-npa} by showing similarly that $\Gamma'$ satisfies them, since it is a limit of $\{\widehat{\Gamma}^{k_l}\}$ and all the $\widehat{\Gamma}^{k_l}$ satisfy these relations when $l$ is large enough. 
\end{proof}

\section{Lasserre Hierarchy for Graph Isomorphism}\label{sec:lasserre}

Recall that two graphs are isomorphic if and only if there exists a permutation matrix $P$ such that $A_G P = P A_H$, where $A_G$ and $A_H$ denote the adjacency matrices of $G$ and $H$ respectively. The problem of checking whether or not two graphs $G$ and $H$ are isomorphic can be formulated as the integer program $\iso(G,H)$ defined as follows: 

\begin{align}
    \begin{aligned}
        \sum_{h \in H} X_{g,h} - 1 & = 1 && \text{for all } g\in V(G), \\
        \sum_{g \in G} X_{g,h} - 1 & = 1 && \text{for all } h\in V(H), \\
        X_{g,h}X_{g',h'} & = 0 && \parbox[t]{6cm}{for all $g,g' \in V(G)$, $h,h' \in V(H)$  such that $\rel_G(g,g') \neq \rel_H(h,h')$.}
    \end{aligned}
\end{align}

where the variables $\{X_{g,h}\}$ are allowed take the values $0,1$. An element $\{g_1h_1, \dots , g_lh_l\}$ is said to be a \emph{partial isomorphism} if $\rel(g_i, g_j) = \rel(h_i, h_j)$ for all $i, j \in [l]$.  We can now consider the Lasserre hierarchy of SDP relaxations for $\iso(G,H)$. We present the version used in \cite{roberson-seppelt-arxiv}. 

\begin{definition}
    Let $k \geq 1$. The level-$k$ Lasserre relaxation for graph isomorphism has variables $y_I$ ranging over $\mathbb{R}$ for $I \in \binom{V(G)\times V(H)} {\leq 2k} $. The constraints are as follows: 
    \begin{align}
        \begin{aligned}
            M_t(y) \coloneqq (y_{I \cup J})_{I, J \in \binom{V(G) \times V(H)}{\leq t}}  &\succeq 0,&& \\
			\sum_{h \in V(H)} y_{I \cup \{gh\}} &= y_{I} &&\text{for all } I \text{ s.t. } \card{I} \leq 2t-2 \text{ and all } g \in V(G), \\
			\sum_{g \in V(G)} y_{I \cup \{gh\}} &= y_{I} && \text{for all } I \text{ s.t. } \card{I} \leq 2t-2 \text{ and all } h \in V(H),  \\
			y_I & = 0 && \text{if } I \text{ s.t.\@ } \card{I} \leq 2t \text{ is not a partial isomorphism,}\\
			y_{\emptyset}& = 1. && \label{lassere5}
        \end{aligned}
    \end{align}
\end{definition}

In \cite{roberson-seppelt-arxiv}, for each $k \in \mathbb{N}$, a class of $(k,k)$-bilabelled graphs $\calL_k$ was constructed such that for each the $k^{\text{th}}$-level of the Lasserre hierarchy $\iso(G,H)$ is feasible if and only if $G,\ H$ are homomorphism indistinguishable over $\calL_k$. The construction of these graph classes $\calL_k$ begins from \emph{atomic graphs} $\calA_k$ which are defined $(k,k)$-bilabelled graphs $\boldsymbol{F} = (F, \boldsymbol{u}, \boldsymbol{v})$ with all of its vertices labelled. Note that the the set of atomic graphs $\mathcal{A}_k$ is generated under parallel composition by the graphs
		\begin{itemize}
			\item $\boldsymbol{J} \coloneqq (J, (1,\dots, k), (k+1, \dots, 2k))$ with $V(J) = [2k]$, $E(J) = \emptyset$,
			\item $\boldsymbol{A}^{ij} \coloneqq (A^{ij}, (1,\dots, k), (k+1, \dots, 2k))$ with $V(A^{ij}) = [2k]$, $E(A^{ij}) = \{ij\}$ for $1 \leq i < j \leq 2k$,
			\item $\boldsymbol{I}^{ij}$ for $1 \leq i < j \leq 2k$ which is obtained from $\boldsymbol{A}^{ij}$ by contracting and removing the edge $ij$. 
		\end{itemize}

  We dedicate this part of the paper to motivate the viewpoint that the NPA hierarchy is a noncommutative generalisation of the Lasserre hierarchy. This is indeed well known in the literature (see \cite{NPA-2, cstar-hierarchy} for example), but we include this for the sake of completion. 

  One can construct an NPA like hierarchy of SDP relaxations of the problem of deciding if two graphs are isomorphic. If we make use of the characterization given in \cref{prop:perf-strat-graph}, a certificate $\cal{R}$ for the $k^{\text{th}}$-level of the NPA hierarchy constructed from a set of commuting projections also satisfies the following condition:
  
  \begin{equation}\label{eq:lasserre-npa}
      \calR_{g_1h_1 \dots g_kh_k, g_{k+1}h_{k+1}\dots g_{2k}h_{2k} } = \calR_{g_{\sigma(1)}h_{\sigma(1)}\dots g_{\sigma(k)}h_{\sigma(k)}, g_{\sigma(k+1)}h_{\sigma(k+1)}\dots g_{\sigma(2k)}h_{\sigma(2k)}} \text{ for all } \sigma \in \mathbb{S}_{2k}. 
  \end{equation} 

  Now, it is not too difficult to see that the certificates for the $k^{\text{th}}$-level of the NPA hierarchy for the $(G,H)$-isomorphism game satisfying~\eqref{eq:lasserre-npa} are in bijective correspondence with certificates for the $k^{\text{th}}$-level of the Lasserre hierarchy. In particular, this implies that the feasibility of the $k^{\text{th}}$-level of the Lasserre hierarchy also implies the feasibility of the $k^{\text{th}}$-level of the NPA hierarchy. Hence, we should expect that $\calP_k \subseteq \calL_k$ for all $k \in \mathbb{N}$ which is indeed the case as we have seen. 

\end{document}